\documentclass[conference]{IEEEtran}

\usepackage{amssymb,amsmath,amsthm,amsfonts}
\usepackage{graphicx}
\usepackage{xcolor}
\usepackage{halloweenmath}
\usepackage{listings}
\usepackage{lstlinebgrd}
\usepackage{xspace}
\usepackage{tcolorbox}
\usepackage{hyperref}
\usepackage{cleveref}
\usepackage{subcaption}
\usepackage{declarations}
\usepackage{mathpartir}
\usepackage{stmaryrd}
\usepackage{multirow}
\usepackage{booktabs}
\usepackage{tabularray}
\usepackage[normalem]{ulem}

\usepackage[available,functional,reproduced]{ndssbadges}

\ifCLASSINFOpdf
\else
\fi

\crefname{equation}{}{}
\Crefname{equation}{}{}

\Crefname{figure}{Figure}{Figures}
\Crefname{table}{Table}{Tables}
\Crefname{theorem}{Theorem}{Theorems}
\Crefname{lemma}{Lemma}{Lemmas}
\Crefname{definition}{Definition}{Definitions}
\Crefname{remark}{Remark}{Remarks}

\hypersetup{
    colorlinks=true,
    citecolor=blue,
    urlcolor=black,
}

\lstdefinestyle{codestyle}{
  basicstyle=\footnotesize\ttfamily,
  numberstyle=\tiny,
  numbers=left,
  float,
  xleftmargin=2em,
  frame=lines,
  breaklines=true,
  tabsize=2,
  escapeinside={<@}{@>},
  captionpos=t,
  abovecaptionskip=-\medskipamount,
  language=C
}
\lstdefinestyle{nolinenumberstyle}{
  basicstyle=\footnotesize\ttfamily,
  numbers=none,
  xleftmargin=0.35em,
  breaklines=true,
  tabsize=2,
  captionpos=t,
  abovecaptionskip=-\medskipamount,
  language=C++
}
\newcommand{\codevar}[1]{\texttt{#1}}
\newcommand{\codefn}[1]{\texttt{#1}}

\newcommand{\authnote}[2]{({\textcolor{blue}{#1}: \textcolor{red}{#2}})}

\newcommand {\changes}[1]{#1}   

\newcommand{\danfeng}[1]{\authnote{Danfeng}{#1}}

\newcommand{\analysis}{\textsc{\framework}\xspace}
\newcommand{\dualbaggy}{\textsc{DualBaggy}\xspace}
\newcommand{\dualdp}{\textsc{DualDeltaPointers}\xspace}

\newcommand{\framework}{\textsc{DualGuard}\xspace}
\newcommand{\frameworkbaggy}{\textsc{DualGuard-\dualbaggy}\xspace}

\newcommand{\frameworksfislh}{\textsc{DualGuard-{\safestack}SLH}\xspace}

\newcommand{\frameworkdp}{\textsc{DualGuard-\dualdp}\xspace}

\newcommand{\baselinebaggy}{\textsc{Baggy+SLH}\xspace}
\newcommand{\baselinesfi}{\textsc{{\safestack}+SLH}\xspace}
\newcommand{\baselinedp}{\textsc{{DeltaPointers}+SLH}\xspace}

\newcommand{\memphaseone}{\textsc{\analysis}$_{\seq}$\xspace}
\newcommand{\memphasetwo}{\textsc{\analysis}$_{\mathghost}$\xspace}

\newcommand{\lightslh}{LightSLH\xspace}
\newcommand{\safestack}{SafeStack\xspace}

\newcommand{\dpselslh}{\textsc{DualDeltaPointers+SLH (Selective)}\xspace}

\newcommand{\baggyselslh}{\textsc{DualBaggy+SLH (Selective)}\xspace}

\newcommand{\sfiselslh}{\textsc{{\safestack}+SLH (Selective)}\xspace}

\newcommand{\True}{\cod{true}}
\newcommand{\False}{\cod{false}}
\newcommand{\Skip}{\cod{skip}}

\newcommand{\while}[2]{\ensuremath{\cod{while}~(#1)~\cod{do}~#2}}

\newcommand{\range}{\mathcal{R}}

\begin{document}

\title{\framework: Comprehensive and Efficient Memory Safety across Sequential and \\ Speculative Execution}

\author{
  \makebox[\textwidth][c]{%
    \begin{tabular}{@{}c@{}}
      Sixuan Dang \\
      Duke University \\
      sixuan.dang@duke.edu
    \end{tabular}\hfill
    \begin{tabular}{@{}c@{}}
      Quan Zhou \\
      Google \\
      quanzhour@google.com
    \end{tabular}\hfill
    \begin{tabular}{@{}c@{}}
      Kaiming Huang \\
      Penn State University \\
      kzh529@psu.edu
    \end{tabular}\hfill
    \begin{tabular}{@{}c@{}}
      Jianan Zhu \\
      UC Riverside \\
      jzhu218@ucr.edu
    \end{tabular}%
  }\\[3ex]
  \makebox[\textwidth][c]{%
    \begin{tabular}{@{}c@{}}
      Anitha Gollamudi \\
      University of Massachusetts Lowell \\
      anitha\_gollamudi@uml.edu
    \end{tabular}\hspace{2.5em}
    \begin{tabular}{@{}c@{}}
      Trent Jaeger \\
      UC Riverside \\
      trentj@ucr.edu
    \end{tabular}\hspace{2.5em}
    \begin{tabular}{@{}c@{}}
      Danfeng Zhang \\
      Duke University \\
      danfeng.zhang@duke.edu
    \end{tabular}%
  }
}


\IEEEoverridecommandlockouts
\makeatletter\def\@IEEEpubidpullup{6.5\baselineskip}\makeatother
\IEEEpubid{\parbox{\columnwidth}{
		Network and Distributed System Security (NDSS) Symposium 2027\\
		22--26 March 2027, Seoul, Republic of Korea\\
		ISBN 978-1-970672-09-1\\  
		https://dx.doi.org/10.14722/ndss.2027.232195\\
		www.ndss-symposium.org
}
\hspace{\columnsep}\makebox[\columnwidth]{}}

\maketitle

\begin{abstract}
Although decades of research have yielded sophisticated mechanisms for enforcing memory safety, speculative execution can subvert these protections through transient memory operations. Conversely, defenses designed specifically for speculative safety, such as Speculative Load Hardening (SLH), cannot enforce memory safety during non-speculative execution. A naive combination of defenses from both domains raises two fundamental challenges: (1) which combination of otherwise fragmented defenses can ensure dual memory safety under speculative and non-speculative executions? (2) whether such combinations introduce redundant protections that incur unnecessary performance overhead?

We present \framework, a unified framework for dual memory protection under non-speculative and speculative executions. \framework systematically analyzes memory behaviors, classifying them into four categories, and further identifies the minimal composition requirements for dual memory safety, which allow us to strategically combine existing mitigations and design new ones to avoid redundant or excessive protections. We implement \framework through a combination of static analysis and runtime enforcement within the LLVM infrastructure. Our evaluation demonstrates that \framework reduces overhead by 3\%–72\% compared to naive compositions of fragmented defenses. In some cases, it even incurs lower overhead than approaches that guarantee memory safety solely under either non-speculative or speculative execution.
\end{abstract}


%
\IEEEpeerreviewmaketitle

\section{Introduction}

Memory safety has long been recognized as a foundational security threat in system security~\cite{cowan2000buffer}. Many defense schemes, ranging from bounds checking \cite{akritidis2009baggy,jones1997backwards,dhurjati2006backwards,kroes2018delta} and sandboxing~\cite{kuznetzov2018safestack,huang2022taming} to compiler and hardware-aided protections~\cite{jim2002cyclone,nagarakatte2013hardware,necula2005ccured,nagarakatte2012watchdog,saileshwar2022heapcheck}, have been developed to prevent memory errors like out-of-bound read (e.g., Heartbleed~\cite{durumeric2014matter}) and write (e.g., Morris worm~\cite{orman2003morris}). 
To reduce the performance overhead of exhaustive checks associated with those schemes, several defense schemes utilize static program analysis to identify and remove redundant checks~\cite{simpson2013memsafe,huang2024top,huang2024comprehensive}: they place checks only at operations whose memory-safety cannot be proven. To further reduce overhead, QuickSafe~\cite{quicksafe} isolates memory objects associated with a known exploit from the rest of the program. These schemes ensure that programs under sequential execution (i.e., architecturally committed or non-speculative) adhere to memory safety constraints. 


However, such memory-safety assurance is undermined by advances in modern processors, namely, {\em speculative execution}. Originally introduced to improve hardware performance through branch prediction and out-of-order execution, speculation has been proven to be an insidious source of transient execution vulnerabilities. Attacks like Spectre~\cite{Kocher2018spectre} and Meltdown~\cite{Lipp2018meltdown} exploit instructions that are speculatively executed to induce new memory errors that are absent in sequential execution. Consider the following code snippet (known as an instance of Spectre attack variant 1):
\vspace{-1ex}
\begin{lstlisting}[numbers=none, frame=none]
if (x < A_size) { y = B[A[x] * 4096] };
\end{lstlisting}
\vspace{-2ex}
The code snippet has no memory error during sequential execution, as $\cod{x}$ is always bounded by $\cod{A\_size}$. But a \emph{speculative memory error} which reads beyond A's bound can occur when the CPU mispredicts the branch outcome and executes $\cod{A[x]}$ speculatively. Although these speculative reads never affect architectural state (like registers and RAM), they can change microarchitectural state such as caches (e.g., which element of $\cod{B}$ is fetched into cache), which can reveal the value accessed speculatively. In light of Spectre attacks, more recent works have revisited memory safety in the context of speculative execution and revealed more attack surfaces due to speculative memory errors~\cite{mambretti2021bypassing,goktas2020speculative,ravichandran2022pacman,kim2025tiktag}.

Despite a substantial body of work on memory safety in the context of sequential execution, and growing interest in its treatment under speculative execution, prior work often implicitly assumes that memory safety mechanisms designed for sequential execution can be directly composed with speculative defenses.
However, as we argue in this paper, this “mitigation-upon-mitigation”
approach 
result in a fragmented protection landscape characterized by two primary limitations:
\begin{itemize}
    \item A naive composition of defenses from these two domains does not necessarily yield both sequential and speculative memory safety, which we term \emph{dual memory safety} in the paper. For example, applying branch-based memory-safety mechanisms like Baggy~\cite{akritidis2009baggy} either before or after SLH-based defenses remains insecure under speculative execution, as they introduce additional memory safety issues during speculative execution (\Cref{sec:motivation}).

    \item A naive combination of defenses from each category can also incur unnecessary overhead. Consider the code below, where a mask-based defense for sequential execution is applied to safeguard accesses to arrays $\cod{A}$ and $\cod{B}$: 
\begin{lstlisting}[numbers=none, frame=none]
if (x < 8) { y = B[A[x mod 8]*4096 mod 8] };
\end{lstlisting}
\vspace{-1ex}
where arrays $\cod{A}$ and $\cod{B}$ have size of 8. The code is already safe under speculative execution. However, speculative defenses (including those that selectively mitigate speculative memory errors~\cite{shivakumar2023selslh,baumann2025fslh,zhu2025lightslh}) still adds redundant protections that incur unnecessary performance overhead.
\end{itemize}

This paper presents the \emph{first} systematic study of the interaction between sequential and speculative \changes{spatial} memory errors: it develops both the theoretical foundations for reasoning about dual memory safety, and practical safety mechanisms that are \emph{efficient and dual memory safe}. In particular, we tackle the following two critical research questions:
\begin{itemize}
    \item 
Q1: Which combinations of otherwise fragmented defenses for sequential and speculative memory safety can ensure dual memory safety?

\item
Q2: How do we achieve dual memory safety without introducing redundant protections that incur unnecessary performance overhead?
\end{itemize}

Our work builds on a rigorous formal framework (\Cref{sec:formalization}) that associates each program location with a \emph{violation region}, defined as the set of out-of-bounds memory accesses that may originate from that location. We further refine this notion by classifying vulnerable accesses within each violation region as sequential-only, speculative-only, or both. This classification isolates the minimal conditions required to prevent each class of memory error.
 
To address Q1, we model \changes{spatial} memory-safety mechanisms as rewriting systems whose goal is to eliminate the violation regions associated with each command. We then identify a key property necessary to achieve dual safety, which we term \emph{safety preservation}. Informally, a sequential (resp. speculative) memory-safety mechanism is safety-preserving if it does not introduce new speculative (resp. sequential) violations. We formally prove that, given two sound memory-safety mechanisms drawn from each domain, applying a safety-preserving mechanism \emph{after} the other always yields dual memory safety. This constitutes the first theoretical contribution of the paper.

To address Q2, we further establish that when both the sequential and speculative mechanisms are sound and safety-preserving, a selective composition strategy becomes feasible. Specifically, the sequential (respectively, speculative) memory-safety mechanism can be safely omitted at program locations that are initially free of sequential (respectively, speculative) memory violations when the corresponding mechanism is applied subsequently. This result enables more efficient compositions while preserving dual memory safety, and constitutes the second theoretical contribution of this work.

Building on these theoretical insights, we analyze the safety properties of a wide variety of mechanisms from the sequential and speculative domains, including Baggy bounds~\cite{akritidis2009baggy}, Delta Pointers~\cite{kroes2018delta}, SafeStack~\cite{kuznetzov2018safestack}, SLH~\cite{slh}, SpecASan~\cite{ganjisaffar2025specasan}, and CHERI~\cite{woodruff2014cheri}. We show that while SafeStack and SLH are safety-preserving, Baggy and Delta Pointers are not (either due to additional branches introduced by code rewriting or underlying hardware enforcement that is unsafe under speculative execution). As a consequence, it is \emph{unsafe} to apply Baggy and Delta Pointers \emph{after} speculative defenses. To the best of our knowledge, this is the \emph{first systematic study} of existing mechanisms for dual memory safety. Moreover, we develop novel variants of Baggy and Delta Pointers, termed \dualbaggy and \dualdp, respectively, which we prove to be dual memory safe by themselves.

Finally, we implement a dual memory safety tool \framework (\Cref{sec:analysis}). \framework employs a static program analysis to systematically identify sources of sequential and speculative \changes{spatial} memory violations at compile time and then applies dual-memory-safe mechanisms, including the novel \dualbaggy and \dualdp mechanisms, selectively to each class of violation regions, thereby eliminating redundant checks. We evaluate \framework on standard security and performance benchmarks, as well as real-world cryptographic libraries (\Cref{sec:evaluation}). The results demonstrate that 
\framework reduces overhead by 3\%–72\% compared to naive compositions of fragmented defenses. In some cases, it even incurs lower overhead than approaches that guarantee memory safety solely under either sequential or speculative execution.



In summary, this work makes the following contributions:
\begin{itemize}
    \item We develop a rigorous formal framework (\Cref{sec:formalization}) that reasons about the interaction between sequential and speculative \changes{spatial} memory errors, as well as key properties that enable safe composition (\Cref{sec:compositon}).

    \item We present the first systematic analysis of dual memory safety across existing sequential and speculative defenses (\Cref{sec:properties}) and develop novel dual-memory-safe variants of Baggy Bounds and Delta Pointers (\Cref{sec:dualbaggy}).
    
    \item  We implement \framework (\Cref{sec:analysis}) that automatically hardens source code to enforce dual memory safety. We evaluate \framework using an extensive benchmark suite of programs (\Cref{sec:evaluation}).
\end{itemize}
\section{Background and Threat Model}


\subsection{Memory Safety}
\label{sec:memsafe}

Violations of memory safety are commonly divided into two categories: spatial and temporal errors.

\textbf{Spatial memory errors} occur when a memory access reads or writes to an address outside of the object's allocated memory bounds. Spatial memory errors remain one of the most prevalent and reliably exploitable classes of software vulnerabilities, which can lead to threats such as buffer overflows and Heartbleed~\cite{durumeric2014matter,orman2003morris}.

\textbf{Temporal memory errors} occur when a memory access reads or writes to an address at a wrong time (e.g., when the address has been deallocated and repurposed). Temporal memory errors can lead to use-after-free bugs, using uninitialized values, and so on.

In this paper, we focus on spatial memory safety, both due to its importance for security, and the fact that it is an important component that enables Spectre-like attacks. Many widely adopted memory safety mechanisms are designed for spatial memory safety as well. For example, mechanisms like bounds checking~\cite{akritidis2009baggy,jones1997backwards,dhurjati2006backwards,kroes2018delta}, software fault isolation (SFI)~\cite{kuznetzov2018safestack, huang2022taming}, and hardware-assisted enforcement~\cite{liljestrand2019pac, partap2022memory}, detect and prevent illegal pointer dereference during non-speculative execution~\cite{vintila2025evaluating}. 




\subsection{Memory Errors under Spectre v1}
\label{sec:spectre}

Among various Spectre and Meltdown attack variants~\cite{canella2019systematic}, a notable class targets the branch prediction unit to circumvent branch-based bounds checking. These attacks, known as ``Spectre v1'' attacks, typically utilize \emph{full gadgets} that speculatively access out-of-bounds memory and subsequently leak it through microarchitectural side channels. For instance, the well-known Spectre attack~\cite{Kocher2018spectre} employs gadgets such as
\lstinline|if (x < A_size) { y = B[A[x] * 4096];} |
to speculatively load memory beyond the bounds of array $\cod{A}$. The resulting access to array $\cod{B}$ modifies a specific cache line that encodes the out-of-bounds value into the cache state. Since the processor does not rollback changes into cache state, an attacker can subsequently infer the secret data via a cache side channel. 

While a full gadget combines an out-of-bound memory access and a cache side channel to reveal secret data, other speculative attacks also exploits \emph{half gadgets}, which only consist of out-of-bound memory accesses. For example, half gadgets 
can be coupled with other microarchitectural attacks like Rowhammer~\cite{kim2014flipping} to reveal the out-of-bounds value~\cite{di2025half}, where half gadgets are used to monitor physical memory activity and guide fault injection. 
In addition, prior works~\cite{bhattacharyya2020specrop, mambretti2021bypassing} utilize out-of-bounds memory writes to launch speculative return-oriented programming (ROP) attacks.  

\subsection{Threat Model}


\framework aims to provide \changes{spatial} dual memory safety across sequential and speculative execution. We adopt a threat model that is consistent with prior research on both memory safety under sequential execution~\cite{akritidis2009baggy,gorter2024sticky}, and security countermeasures against Spectre v1 attacks~\cite{cauligi2022sok,shivakumar2023selslh,zhang2023ultimate}.

More specifically, we consider traditional \emph{spatial} memory safety violations, such as out-of-bound pointer dereferences or invalid memory accesses that occur during sequential execution. Spatial memory safety violations are also the root cause of Spectre v1 vulnerabilities. On the other hand, following prior studies of memory safety~\cite{gorter2024sticky,moghadam2024memory}, we leave temporal memory errors  (e.g., use-after-free) for future work, as they are orthogonal to Spectre attacks.  

For speculative execution, we assume an adversary who is capable of fully exploiting speculative program execution, including controlling branch prediction unit to step into an arbitrary branch at a branching instruction, and provide untrusted inputs that affects which memory address is being accessed speculatively. While most prior works~\cite{zhu2025lightslh,zhang2023ultimate} also assume that only a cache-based side channel that can reveal out-of-bounds values (i.e., they only consider full-gadgets), 
we consider a stronger threat model where \emph{all transient memory accesses} might be exploited (e.g., via Rowhammer attacks~\cite{di2025half} or ROP attacks~\cite{bhattacharyya2020specrop, mambretti2021bypassing}), including those where no immediate cache-based side channels exist after the speculative memory access. In other words, our threat model considers both full and half Spectre gadgets (Section~\ref{sec:spectre}).


Finally, our threat model covers Spectre v1 (bounds check bypass) attacks. Other variants, including Spectre v2 (BTB poisoning)~\cite{barberis2022branch}, Spectre-RSB~\cite{koruyeh2018spectre}, Spectre v4 (speculative store bypass)~\cite{cauligi2022sok}, as well as Meltdown attacks~\cite{Lipp2018meltdown}, are out of the scope in this paper. \changes{We discuss possible future directions beyond the scope of this paper in Section~\ref{sec:future_safety}.}




\subsection{Motivation}
\label{sec:motivation}

Next, we demonstrate why a naive composition of fragmented defenses for memory safety does not achieve sound and efficient dual memory safety.  We take Baggy bounds checking~\cite{akritidis2009baggy} as an example of sequential memory-safety mechanism, and SLH~\cite{slh} as a speculative-safety mechanism. 

Consider the following code snippet.
\begin{lstlisting}[numbers=none, frame=none]
if (x < 9) { y = A[x]; }
\end{lstlisting}
Assume the size of array \codevar{A} is 8. Since the code may access \codevar{A[8]} due to the incorrect bounds check, the code is memory unsafe during sequential execution. Moreover, with speculative execution, an attacker can also perform out-of-bounds access by manipulating the branch predictor.

When both Baggy and SLH are available, one can compose them in either order, with SLH applied before Baggy or vice versa. However, we show that both composition orders are unsafe, and a
naive composition results in redundant layers of protection and hence, unnecessary performance overhead.

\paragraph{SLH first, then Baggy} SLH thwarts Spectre-like attacks by (1) introducing data dependency between the branch condition and the access to the array $\codevar{A}$, and (2) masking the index to the array $\codevar{A}$. 
With SLH, the code is hardened as:

\begin{lstlisting}[
    linebackgroundcolor={%
    \ifnum\value{lstnumber}=2
            \color{blue!30}
    \fi
    \ifnum\value{lstnumber}=4
            \color{blue!30}
    \fi
    }
]
if (x < 9) {
  mis = x < 9 ? 0 : -1;
  addr = &A[x];
  hardened_addr = addr | mis;
  y = load(hardened_addr); }
\end{lstlisting}
Here, \codevar{mis} is a dedicated register tracking the mis-speculation status: it stores all-one bits under misprediction, and all-zero bits otherwise. The address of $\cod{A[x]}$ is hardened by operation OR between the computed addresses and \codevar{mis}. Therefore, if the prediction is correct, \codevar{addr | 0} preserves \codevar{addr}; otherwise, every bit is set to 1, resulting in an invalid address reserved by SLH.

Then, Baggy introduces new branches to protect against sequential memory errors: it rewrites the memory access to:
\begin{lstlisting}[
    linebackgroundcolor={%
    \ifnum\value{lstnumber}=2
            \color{blue!30}
    \fi
    \ifnum\value{lstnumber}=3
            \color{blue!30}
    \fi
    \ifnum\value{lstnumber}=4
            \color{blue!30}
    \fi
    }
]
addr = &A[x];
log_size = size_table[A >> log_slot_size];
if((addr XOR A) >> log_size != 0)
    addr = slow_path(A, addr);
y = load(addr);
\end{lstlisting}
Here,
Baggy introduces a new array access at Line 2, which never goes out of bounds. The reason is that the size table it indexes into has a size of $memory\_size/slot\_size$, where $slot\_size$ is a fixed constant. Due to the right shifting operation,
the introduced array access is always spatially safe.



The branch introduced by Baggy at Line 3, however, is not speculatively safe: it may be mispredicted, allowing an attacker to speculatively reach Line 5 despite Baggy’s check at Line 3. 
Moreover, \codefn{slow\_path} itself contains conditional branches to determine whether an error is repairable, which are also speculatively unsafe, and further enlarge the attack surface.


\paragraph{Baggy first, then SLH} When reversing the composition order, the combination may initially appear dual safe, as the branching instructions introduced by Baggy are mitigated by SLH, thereby removing the attack surface discussed above.

However, this composition is still not dual memory safe due to a more subtle vulnerability. Specifically, Baggy relies on \emph{canonical address checks}, a hardware-level safety mechanism in modern CPUs, to reject accesses to invalid addresses.
Recent work has demonstrated that such checks can be bypassed under speculative execution (e.g., CVE-2020-12965)~\cite{amdCVE}. Consequently, the composition of “Baggy then SLH” remains unsafe under speculative execution, as SLH does not enforce protections for canonical address checks.

\paragraph{Performance overhead} Beyond the security concerns above, a naive composition of fragmented memory safety defenses may introduce redundant or excessive protections, particularly given that these mechanisms already incur nontrivial performance overheads when deployed independently. For instance, on a cryptographic library benchmark, we observe that composing off-the-shelf Baggy with SLH results in up to 10X slowdown (\Cref{sec:eva-dual}).

In summary, naive composition of existing defense mechanisms for sequential and speculative memory safety does not necessarily provide the desired safety guarantees. Moreover, the resulting performance overhead is often prohibitive, limiting the practicality of such approaches for real world software.

\section{Related Work}

\subsection{Memory Safety in Spectre Era}
\label{sec:specmemsafe}
In light of Spectre attacks, several works have revisited memory safety in the context of speculative execution: they analyze how speculative behaviors can invalidate existing memory safety mitigations~\cite{mambretti2021bypassing,goktas2020speculative,ravichandran2022pacman,kim2025tiktag, ganjisaffar2025specasan}. For example, PACMAN~\cite{ravichandran2022pacman} and TikTag~\cite{kim2025tiktag} show that speculative execution could break ARM pointer authentication (PA) and memory tagging extension (MTE), leading to control flow hijacking. BlindSide~\cite{goktas2020speculative} uses memory errors to break Address Space Layout Randomization (ASLR) mechanisms via speculatively probing randomized memory layout without crashing the system. SPEAR~\cite{mambretti2021bypassing} aims to bypass stack canary, CFI and bounds check in GO~\cite{golang}. These works focus primarily on how speculation undermines prior memory-safety defenses. However, they ignore the fact that speculation itself could introduce extra memory errors that are not exposed during sequential execution. As a consequence, they often suggest layering speculative safety mechanisms (\Cref{sec:related-spec}) on top of traditional ones, creating potentially unsafe and redundant protection stacks~\cite{christou2024eclipse, ganjisaffar2025specasan} that \framework aims to identify and remove.

\framework re-examines the interaction between memory safety and speculative execution from a unified perspective. Rather than treating speculative safety as a retrofit to existing memory safety systems, it integrates both domains into a systematic view of when and where protection is truly necessary.



\subsection{Defenses against Spectre V1 Attacks}
\label{sec:related-spec}

Since the seminal work of Spectre~\cite{Kocher2018spectre}, both academia and industry have proposed a broad spectrum of defenses~\cite{canella2019systematic}. Most related to this paper are software-based defenses, which fall into three categories.

    
    Architectural primitives such as \codefn{LFENCE}~\cite{intel2018lfence}, \codefn{IBRS}~\cite{intel2018ibrs}, or ARM's \codefn{CSDB}~\cite{arm2025csdb} serialize instructions to disable speculation completely. While effective when carefully placed, pervasive use of these primitives incurs prohibitive overhead.

    Observing the prohibitive overhead of disabling speculation completely, defense mechanisms like ones based on LAHF instruction~\cite{Oleksenko:2018} and Speculative Load Hardening (SLH)~\cite{slh} introduce artificial data dependencies between branch conditions and subsequent memory operations to ensure that speculative accesses cannot touch sensitive data before branch resolve. Notably, SLH~\cite{slh} was implemented in LLVM since version 8. To reduce the overhead of SLH, several optimized variants have been proposed, limiting instrumentation to necessary control flows only~\cite{baumann2025fslh,zhu2025lightslh,shivakumar2023selslh}. While these selective and lightweight variants share the same objective of removing unnecessary protection as \framework, they do not protect against sequential memory errors.

\section{Formalizing Memory Errors}
\label{sec:formalization}

Our goal is to provide \emph{dual memory safety}. A central challenge in achieving this goal is the need to precisely reason about both sequential and speculative memory behaviors, as well as their interactions. To address this challenge, we first define a small imperative language together with a formal semantics that captures both sequential and speculative memory errors. We then introduce a classification of memory behaviors into four categories, presenting each both informally and formally. This categorization provides a theoretical foundation for reasoning about dual memory safety (Section~\ref{sec:defense}).

\subsection{Language Syntax}

\begin{figure}
    \centering \small
    \begin{equation*}
        \begin{aligned}
            & \textbf{Expressions} \\
            & e ::= n \mid x \mid a[e] \mid e_1 \otimes e_2 \\
            & \textbf{Boolean Expressions} &\\
            & b ::= \True \mid \False \mid e_1 \odot e_2 \mid \lnot b \mid b_1 \lor b_2 \mid b_1 \land b_2 \\
            & \textbf{Command} &\\
            & \cmd ::= \Skip \mid x := e \mid a[e_1] := e_2  \mid  \cmd_1;\, \cmd_2 \mid \\
            & \ifcmd{b}{\cmd_1}{\cmd_2}  \mid \while{b}{\cmd}
        \end{aligned}
    \end{equation*}
    \vspace{-2ex}
    \caption{Language syntax. $\otimes$ and $\odot$ represent binary arithmetic and comparison operations respectively.}
    \label{fig:syntax}
\end{figure}

Following prior work~\cite{brotzman2021specsafe}, we consider a simple imperative language whose syntax is defined in \Cref{fig:syntax}. An expression can be a constant $n$, a variable $x$, a pointer arithmetic $a[e]$ (where $a[e]$ is an address computed with a base address $a$ and an offset $e$), or a binary arithmetic operation $\otimes$ on two expressions. 
The language captures assignment, branching, sequential composition, and memory operations in the form of $x := e$ and $a[e_1] := e_2$. Since $e$ by definition may contain a memory access operation $a[e']$, both memory loads and stores are modeled in the language syntax, and it supports nested array access like $\cod{B[A[x] * 4096]}$. We select a source-level language for formalization because of its readability and simplicity. Our code analysis is conducted on lower-level intermediate representation code, as detailed in \Cref{sec:evaluation}.

\subsection{Semantics}
\label{sec:semantics}

\textit{Memory model.} We model the memory state $\mem$ as a mapping from memory addresses to the values stored at those addresses. Additionally, we define a global range $\range$ as a mapping from memory objects $\obj$ to the allocated memory regions $\range(\obj)$, which are sets of continuous memory addresses that are considered safe for $\obj$ to access.  For instance, a fine-grained object-level memory safety mechanism, such as Baggy~\cite{akritidis2009baggy}, defines the range of an object as $\range(\obj) = [\base{\obj}, \base{\obj} + \size{\obj})$, where $\base{\obj}$ represents the base address of the object and $\size{\obj}$ denotes its allocated size. A coarse-grained memory safety mechanism, like SFI~\cite{kuznetzov2018safestack}, specifies the range of an object as $\range(\obj) = [\base{S}, \base{S} + \size{S})$, where $S$ is the sandbox that the object resides within, and $\base{S}$ and $\size{S}$ represent the base address and size of the sandbox $S$, respectively.

\textit{Program state.} Following~\cite{brotzman2021specsafe,guarnieri2020spectector}, the speculative-aware language semantics is parameterized on a global \emph{prediction oracle}, $\oracle$, a partial function that takes the location of a branch command and returns a queue of predicted branch outcomes (i.e., a queue of truth values). Moreover, it assumes a fixed transaction window size $w$, modeling the size of the reordering buffer in processor. Speculatively executed instructions are executed under a \emph{speculative transaction}~\cite{guarnieri2020spectector}, which is ultimately either committed (when the prediction is correct) or rolled back (when the prediction is incorrect)\footnote{Each transaction tracks information like the memory state and command to resume in case of a rollback. Since the rest of the paper does not depend on its precise definition, we leave the full definition in the Appendix.}. 

We define program state as a tuple $\configthree{\mem}{\vec{\micro}}{\cmd}$ that consists of a memory state $\mem$, a sequence of (potentially nested) speculative transactions $\vec{\micro}$ as well as $\cmd$, the remaining program to be executed. A single-step evaluation rule has the form of:
\[\configthree{\mem}{\vec{\micro}}{\cmd} \evalsto{\event} \configthree{\mem'}{\vec{\micro}'}{\cmd'}\]
where event $\event$ consists of (1) $\start{i}$, $\rollback{i}$, and $\commit{i}$, which denote the start, rollback, and commit of the $i$th transaction respectively, and (2) $\load{\location}{a[n]}$ and $\store{\location}{a[n]}$ indicate the loading and storing of memory address $a[n]$ at program location $\location$. Given a global prediction oracle and a fixed window size, the semantics is deterministic and yields a single execution trace that captures both the sequential and speculative behaviors of the program. 
Since the semantics mostly follows prior work~\cite{brotzman2021specsafe}, we omit 
the evaluation rules here and formalize them in the Appendix.



\textit{Speculative traces and trace projection.} An execution of the entire program generates a trace of events, which consist of both memory operations and speculative transaction labels (i.e., $\start{i}$, $\rollback{i}$, and $\commit{i}$) around those operations. Let $\cmd$ represent a program and $\oracle$ represent an oracle that describes the behavior of a branch predictor. We write $\configthree{\mem}{\vec{\micro}}{\cmd} \hookrightarrow_{\oracle} \tau$ to represent that a trace $\tau$ is generated from an initial state $\configthree{\mem}{\vec{\micro}}{\cmd}$ given an oracle $\oracle$.

For instance, consider a code snippet:
\begin{lstlisting}
if (x < 8)
  y = A[x];
else
  y = A[0];
\end{lstlisting}
Assume we have $\size{\texttt{A}}=8$ and $\texttt{x} = 10$ in the initial program state $\state_0$. With an oracle that always mispredicts a branch~\cite{guarnieri2020spectector}, the speculative trace of this program is
\begin{equation*}
\trace = \start{0}~\load{2}{A[10]}~\rollback{0}~\load{4}{A[0]}
\end{equation*}
From the trace, it is evident that $\cod{A[10]}$ causes a memory error, but only under speculative execution. 

To analyze the sequential and speculative behaviors of a program’s execution, as well as their interactions, we extract two projections from $\trace$: sequential and speculative projections.

\begin{definition}[Sequential Projection] The sequential projection of a trace $\trace$, written as $\projection{\seq}$, is the longest subsequence of trace $\trace$, excluding all rolled-back speculative transactions, i.e., all operations enclosed by a matching pair of $\start{i}$ and $\rollback{i}$. 
\end{definition}

In the example above, the sequential projection is computed by removing the speculative transaction between $\start{0}$ and $\rollback{0}$, resulting in $\projection{\seq} = \load{4}{A[0]}$. The sequential projection corresponds exactly to the trace that is produced by executing the program under sequential semantics.

\begin{definition}[Speculative Projection] 
A speculative projection of a trace $\trace$, written as $\projection{\mathghost}$, is the complement of sequential projection $\projection{\seq}$, i.e., the longest subsequence of trace $\trace$ with operations enclosed by a matching pair of $\start{i}$ and $\rollback{i}$. 
\end{definition}

In the example above, we have $\projection{\mathghost} = \load{2}{A[10]}$, indicating that the memory error due to out-of-bound access $\cod{A[10]}$ is only executed under a mis-prediction.

Note that the trace projections above only classify \emph{each memory access} to be either sequential or speculative. However, for developing memory safety mechanisms, we are more interested in whether \emph{each program location} can produce sequential memory errors, speculative memory errors, or both. We develop such categorizations on memory locations next.

\subsection{Categorization}
\label{sec:categorization}

To reason about the memory-related behavior of each program location, we write $\exec{\cmd}$ to denote all possible sequential executions of $\cmd$ (i.e., $\{ \projection{\seq} \mid \exists \mem, \oracle.~\configthree{\mem}{\emptyset}{\cmd} \hookrightarrow_{\oracle} \trace \}$), and $\specexec{\cmd}$ to denote its complement (i.e., executions of $\cmd$ under speculation). Note that by definition, both $\exec{\cmd}$ and $\specexec{\cmd}$ consist of a set of memory accesses in the form of $\load{\location}{\obj[i]}$ or $\store{\location}{\obj[i]}$, where $\location$ is the program location that issues the memory access to $\codevar{\obj[i]}$.

For each program location $\location$, we define its \emph{Access Set} in trace $\trace$ as the set of all memory loads and stores issued by $\location$:
\begin{definition}[Access Set] Given a program location $\location$, its access set ($\accessset$) in a program execution trace $\trace$ is:
\begin{equation}
\label{equ:accset}
\begin{aligned}
    \accessset(\trace, \location) = &\{ \codevar{\obj[i]} \mid \load{\location}{\codevar{\obj[i]}} \in \trace  \} \\
     \cup &\{\codevar{\obj[i]} \mid \store{\location}{\codevar{\obj[i]}} \in \trace  \}  
\end{aligned}
\end{equation}
\end{definition}

Note that access set is defined on any execution trace $\trace$, which consists of both sequential and speculative memory accesses, as well as their combinations. Moreover, if a program location $l$ can legally index into multiple objects (e.g., via pointer aliasing), the definition includes all memory accesses to those objects.

Next, we define memory access that \emph{might} violate memory safety at each program location.

\begin{definition}[Sequential Violation Region] A sequential violation region of a program location $\location$, $\vioregionseq(\location)$, is the set of memory addresses that $\location$ can access out-of-bounds during sequential execution.
\begin{equation*}
\label{equ:vioregion}
\vioregionseq(l) = \bigcup_{\trace \in \exec{\cmd}} \{\codevar{\obj[i]} \mid \codevar{\obj[i]} \in \accessset(\trace, \location) \land \codevar{\obj[i]} \notin \range(\codevar{\obj})\}
\end{equation*}
\end{definition}

\begin{definition}[Speculative Violation Region] A speculative violation region of a program location $\location$, $\vioregionspec(\location)$, is the set of memory addresses that $\location$ can access out-of-bounds during speculative execution.
\begin{equation*}
\label{equ:vioregionspec}
\vioregionspec(l) = \bigcup_{\trace \in \specexec{\cmd}} \{\codevar{\obj[i]} \mid \codevar{\obj[i]} \in \accessset(\trace, \location) \land \codevar{\obj[i]} \notin \range(\codevar{\obj})\}
\end{equation*}
\end{definition}


Next, we formally categorize each program location into one of the four categories based on whether its violation regions are empty. Although the categorization appears straightforward, it serves as a formal foundation for identifying the \emph{minimal requirements} on each program location to achieve dual memory safety, which we elaborate in~\Cref{sec:defense}.

\begin{definition}[Categorization]
\label{def:cat}
Given a program location $\location$, categorization ($\mathcal{C}$) is a mapping that classifies $\location$ into one of the four categories, S$_\seq$S$_{\mathghost}$, U$_\seq$S$_{\mathghost}$, S$_\seq$U$_{\mathghost}$ and U$_\seq$U$_{\mathghost}$, where $S$ stands for safe, and $U$ stands for unsafe.
\end{definition}
\begin{equation}
\label{equ:cat}
\mathcal{C}(l) = 
    \begin{cases}
        \text{S}_\seq \text{S}_{\mathghost}, \vioregionseq(l)=\emptyset \land \vioregionspec(l)=\emptyset \\
        \text{U}_\seq \text{S}_{\mathghost}, \vioregionseq(l) \neq \emptyset \land \vioregionspec(l)=\emptyset \\
        \text{S}_\seq \text{U}_{\mathghost}, \vioregionseq(l)=\emptyset \land \vioregionspec(l) \neq \emptyset \\
        \text{U}_\seq \text{U}_{\mathghost}, \vioregionseq(l) \neq \emptyset \land \vioregionspec(l) \neq \emptyset
    \end{cases}
\end{equation}

\begin{figure}
    \centering
    \includegraphics[width=0.35\linewidth]{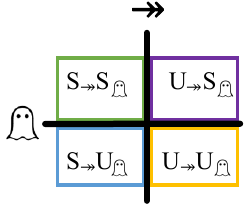}
    \vspace{-2ex}
    \caption{Categorization of memory behaviors}
    \label{fig:category}
\end{figure}

Intuitively, the categorization classifies each memory-accessing location $\location$ into one of four regions shown in \Cref{fig:category}.

\begin{itemize}
    \item S$_\seq$S$_{\mathghost}$: dual memory safe.
    \item U$_\seq$S$_{\mathghost}$: sequentially unsafe but speculatively safe.
    \item S$_\seq$U$_{\mathghost}$: sequentially safe but speculatively unsafe.
    \item U$_\seq$U$_{\mathghost}$: sequentially unsafe and speculatively unsafe.
\end{itemize}

We use the examples in~\Cref{code:categorize} to elaborate on each category. For all examples, assume array \codevar{A} has a size of 8.

\begin{figure}
    \centering
    \begin{subfigure}[b]{0.23\textwidth}
    \centering
    \begin{lstlisting}
y = A[3];
if(i < 8)
    y = A[3];
    \end{lstlisting}
    \vspace{-2ex}
    \caption{Always safe}
    \label{code:cat1}
    \end{subfigure}
\hfill
\begin{subfigure}[b]{0.23\textwidth}
\centering
\begin{lstlisting}
y = A[i];
if(i >= 8)
    y = A[i]
\end{lstlisting}
\vspace{-2ex}
\caption{$\seq$ Unsafe, $\mathghost$ Safe}
\label{code:cat3}
\end{subfigure}

    \begin{subfigure}[b]{0.23\textwidth}
    \centering
    \begin{lstlisting}
i = 10;
if(i < 8)
    y = A[i];
    \end{lstlisting}
    \vspace{-2ex}
    \caption{$\seq$ Safe, $\mathghost$ Unsafe}
    \label{code:cat2}
    \end{subfigure}
\hfill
    \begin{subfigure}[b]{0.23\textwidth}
    \centering
    \begin{lstlisting}
i = 0;
while (i < 10)
    y = A[i];
    i = i + 1;
    \end{lstlisting}
    \vspace{-2ex}
    \caption{Both unsafe}
    \label{code:cat4}
    \end{subfigure}
    \vspace{-1ex}
    \caption{Motivating examples, where the size of array $\codevar{A}$ is 8.}
    \label{code:categorize}
\end{figure}

\textit{Case S$_\seq$S$_{\mathghost}$.}
\Cref{code:cat1} shows an array access using a constant index that is smaller than the array size. Both memory reads at Line 1 and 3, which are under a branch, are always memory safe as the index is set to 3. The intuition can be confirmed with \Cref{def:cat}, as violation regions $\vioregionseq(1)$, $\vioregionspec(1)$, $\vioregionseq(3)$ and $\vioregionspec(3)$ are all empty. Since the code is always safe, no memory-safety defense is required.

\textit{Case U$_\seq$S$_{\mathghost}$.} One instance of U$_\seq$S$_{\mathghost}$ is an out-of-bound access that occurs only under sequential execution, such as when the execution is not under any branches. \Cref{code:cat3} presents a contrived yet more interesting and counterintuitive example where the access is guarded by a branch. 

When non-negative $i$ is smaller than 8, both the sequential and speculative executions read $\codevar{A[i]}$ which is memory safe. However, when $i\geq 8$, say 10, the sequential execution reads $\codevar{A[10]}$ twice, which is out-of-bound, while the speculative execution does nothing after the first access at Line 1. We can see this more clearly with its execution trace under the always-mispredict oracle:
\begin{equation*}
\begin{aligned}
\trace =  &\load{1}{A[10]} \ \start{0} \ \rollback{0} \\
          & \start{1} \ \load{3}{A[10]} \ \commit{1}
\end{aligned}
\end{equation*}
For this trace, $\projection{\seq}=\{\load{1}{A[10]}, \load{3}{A[10]}\}$, which are both out of bound, and $\projection{\mathghost}=\emptyset$. Hence, for $l=1$ or $3$, we have $\vioregionseq(l)\neq \emptyset$ and $\vioregionspec(l)=\emptyset$. So both Line 1 and 3 are categorized as U$_\seq$S$_{\mathghost}$.
As for defense, such program locations require sequential memory safety defenses at minimum, given that the defensive instrumentation itself does not introduce additional memory-safety violations (\Cref{sec:specmemsafe}). 

\textit{Case S$_\seq$U$_{\mathghost}$.} This category includes classic Spectre v1 attacks, where out-of-bound memory access only occurs under a misprediction. 
\Cref{code:cat2} shows a code snippet where the array access $\codevar{A[i]}$ is guarded by a bounds check at Line 2. Under sequential semantics, the access at Line 3 is never reached and thus free from memory errors. But with speculative execution, Line 3 can be speculatively executed and access $\codevar{A[10]}$, which is out-of-bounds.  We can see this more clearly with its execution trace:
\begin{equation*}
\trace = \start{0} \ \load{3}{A[10]} \ \rollback{0}
\end{equation*}
For this trace, we have $\projection{\seq}=\emptyset$ and $\projection{\mathghost}=\load{3}{A[10]}$, which is out of bound. Hence, $\vioregionseq(3)=\emptyset$ but $\vioregionspec(3)\neq \emptyset$.
As the memory error only occurs during speculative execution, defenses against Spectre V1 attacks, such as adding $\codefn{LFENCE}$ or applying SLH, suffice to remove the speculative memory error, given that they do not introduce additional memory-safety violations (\Cref{sec:specmemsafe}).




\textit{Case U$_\seq$U$_{\mathghost}$.} \Cref{code:cat4} shows an example that is both sequentially unsafe and speculatively unsafe. Under sequential execution, the memory access at Line 3 will access \codevar{A[8]} and \codevar{A[9]} out-of-bounds in the last two loop iterations because of the incorrect bounds check at Line 2. Moreover, under misprediction, iterations after that might access \codevar{A[i]} where $i \ge 10$ under misprediction, which leads to speculative memory error as well. We can see this more clearly with the following execution trace where only the last prediction when $i=10$ has a misprediction:
\begin{equation*}
\begin{aligned}
    \trace = &\start{0} \ \load{3}{A[0]} \ \commit{0} \\ 
    &... \ \\
    &\start{9} \ \load{3}{A[9]} \commit{9} \\ 
    &\start{10} \ \load{3}{A[10]} \ \rollback{10} 
\end{aligned}
\end{equation*}
Hence,  $\projection{\seq}=\{\load{3}{A[0]}, \dots, \load{3}{A[9]}\}$ and $\projection{\mathghost}=\{\load{3}{A[10]}\}$. So $\vioregionseq(3) = \{ \codevar{A[8]}, \codevar{A[9]} \}$ and $\vioregionspec(3) = \{ \codevar{A[10]} \}$. 

\section{Dual Memory Safety}
\label{sec:defense}

One important impact of the categorization of program locations is to enable formal and precise reasoning about dual memory safety, which is formalized as follows.
\begin{definition}[Dual Memory Safety]
\label{def:dualsafe}
A program satisfies dual memory safety if it is free of all memory violation regions: 
\[\forall \location \in \cmd.~\vioregionseq(\location) = \emptyset \land \vioregionspec(\location) = \emptyset\]
\end{definition}

\subsection{Composition Theorem}
\label{sec:compositon}
To ensure dual memory safety, one challenging question is how to combine, or modify existing memory-safe mechanisms with partial protection. As discussed in \Cref{sec:motivation}, a naive composition of existing defense mechanisms does not necessarily provide dual safety. Moreover, the resulting performance overhead can be prohibitive.

To address the challenge, we first model each memory safety mechanism as a program rewriting process, and use $\defenseseq{\location}$ and $\defensespec{\location}$ to denote defense mechanisms for sequential and speculative memory errors respectively. A sound memory-safe mechanism should possess the following property:
\begin{definition}[Sound Memory-Safety Mechanism]
\label{def:soundseq}
\label{def:soundspec}
A sequential memory-safe mechanism is sound if it removes all sequential violation locations:
$\forall \location\in \cmd. ~\vioregionseq(\defenseseq{\location})=\emptyset $.\\
Likewise, a speculative memory-safety mechanism is sound if it removes all speculative violation locations:
$\forall \location\in \cmd. ~\vioregionspec(\defensespec{\location})=\emptyset $.
\end{definition}


To systematically reason about the composition of memory-safe mechanisms, we identify an essential condition that makes the composition of two memory-safety mechanisms dual memory safe. We term this essential property \emph{safety preservation}. 
\begin{definition}[Safety Preservation]
\label{def:safety}
A sequential memory-safety mechanism $\defenseseq{}$ preserves memory safety if it does not introduce new speculative violations: $\forall \location\in \cmd.~ \vioregionspec(\defenseseq{\location}) \subseteq \vioregionspec(\location)$.

Similarly, a speculative memory-safe mechanism $\defensespec{}$ preserves memory safety if
it does not introduce new sequential violations:
$\forall \location\in \cmd.~\vioregionseq(\defensespec{\location}) \subseteq \vioregionseq(\location)$.
\end{definition}

Note that \Cref{def:safety} quantifies over each program location  $\location$ in $\cmd$, the program \emph{prior to rewriting}. Since a given mechanism might instrument multiple instructions at a single location $\location$, the rewritten code, e.g., $\defenseseq{\location}$, is a set of instructions. For notational convenience, we overload $\vioregionspec(\defenseseq{\location})$ to denote the \emph{union} of violation sets associated with all commands in $\defenseseq{\location}$.

Next, we prove the composition theorem.

\begin{theorem}[Composition]
\label{thm:composition}
For any program $\cmd$ and two sound memory-safe mechanisms $\defenseseq{}$ and $\defensespec{}$:
\begin{itemize}
    \item $\defensespec{}$ is safety preserving $\Rightarrow$ $\defensespec{\defenseseq{\cmd}}$ is dual memory safe.

    \item $\defenseseq{}$ is safety preserving $\Rightarrow$ $\defenseseq{\defensespec{\cmd}}$ is dual memory safe.
\end{itemize}
\end{theorem}

\begin{proof}
We prove the first bullet as the other case is symmetric.

\begin{itemize}
    \item For speculative memory safety, we consider an arbitrary  location $\location\in \cmd$. Then $\vioregionspec(\defensespec{\defenseseq{\location}})=\emptyset$ since $\defensespec{}$ is sound. 

    \item For sequential memory safety, we consider an arbitrary  location $\location\in \cmd$. Then $\defenseseq{\location}$ represents the commands \emph{after} applying $\defenseseq{}$ to $\location$. 
    
    Since $\defensespec{}$ is safety preserving, $\forall \location'\in \defenseseq{\location}.~\vioregionseq(\defensespec{\location'}) \subseteq 
    \vioregionseq(\location')$. Moreover, since $\defenseseq{}$ is sound, we have 
    $\vioregionseq(\location')=\emptyset$ for each $\location'$.  Hence, $\bigcup_{\location'\in \defenseseq{\location}} \vioregionseq(\defensespec{\location'}) = \emptyset $, or equivalently, $\vioregionseq(\defensespec{\defenseseq{\location}}) = \emptyset $. Since this holds for any $\location\in \cmd$, $\defensespec{\defenseseq{\cmd}}$ is sequentially memory safe.
\end{itemize}
\end{proof}

\subsection{Systematic Analysis of Memory Safety Mechanisms}
\label{sec:properties}
To demonstrate the practical relevance of \Cref{thm:composition}, we first investigate the safety properties of existing mechanisms, summarized in \Cref{tab:property}. We select a variety of representative memory safety mechanisms, ranging from object-based bounds checking (Baggy~\cite{akritidis2009baggy}), pointer-based bound checking (Delta Pointers~\cite{kroes2018delta}),  sandboxing (SafeStack~\cite{kuznetzov2018safestack}), Spectre defense (SLH~\cite{slh}), to hardware-aided (SpecASan~\cite{ganjisaffar2025specasan} and CHERI~\cite{woodruff2014cheri, fuchs2024safe}) mechanisms. We subsequently show that our newly developed composition theory significantly streamlines the analysis of dual safety.

\textit{Baggy is not safety-preserving.} Baggy introduces branches in the transformed code, which may expand the speculative violation region. Recall the example in \Cref{sec:motivation}, where Line 1 $\cod{y=A[x]}$ is transformed to a branch-based bounds check. Before applying Baggy, 
$\vioregionspec(1) = \emptyset$ as Line 1 is not guarded by any branch. However, with Baggy, it is sequentially sound as $\vioregionseq(\defenseseq{1}) = \emptyset$. However, $\vioregionspec(\defenseseq{1}) = \{\cod{A[8]}, \cod{A[9]}, ...  \}$ due to the branch-based bounds check added by Baggy. Hence, Baggy is not safety-preserving, as $\vioregionspec(\defenseseq{1}) \not\subseteq \vioregionspec(1)$.

\textit{Delta Pointers is not safety-preserving.} Delta Pointers~\cite{kroes2018delta} is a pointer-based defense that divides a 64-bit virtual address into three parts, the most significant bit (MSB), a 31-bit delta tag, and a 32-bit address.

\vspace{-3ex}
\begin{figure}[htp]
    \centering
    \includegraphics[width=0.6\linewidth]{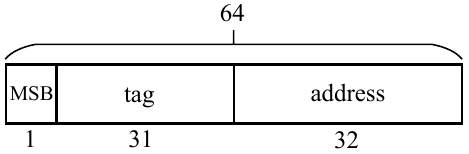}
    \label{fig:address}
\end{figure}
\vspace{-3ex}

The tag bits encode the distance between the current pointer value and the upper bound of the associated object. 
During pointer arithmetic, the tag is updated relative to the address change. When the pointer moves out of bounds, by design, the arithmetic causes the tag to overflow into the MSB, which triggers a hardware exception as modern hardware enforces canonical address checks. 
For an memory access such as \codevar{y = A[x]} (where the size of \codevar{A} is 8), the instrumented code is
\begin{lstlisting}[
    linebackgroundcolor={%
    \ifnum\value{lstnumber}=2
            \color{blue!30}
    \fi
    \ifnum\value{lstnumber}=3
            \color{blue!30}
    \fi
    }
]
addr = &A[x];
new_addr = addr + (x << 32);
masked_addr = new_addr & ADDR_MASK;
y = load(masked_addr);
\end{lstlisting}
where \codevar{ADDR\_MASK} is the mask removing the tag bits.
When \codevar{x $\ge$ 8}, the addition at Line 2 overflows the tag to MSB, resulting a hardware exception at Line 4 since MSB is set.

Unlike Baggy, Delta Pointers does not introduce branch-based software checks. However, it relies on hardware canonical address checks to prevent dereferencing of invalid pointers, which can be bypassed under speculative execution on AMD CPUs~\cite{amdCVE, noncanonical, musaev2021transient}. \changes{Recent research~\cite{hertogh2024leaky} also found that Intel CPUs can be affected, but by Spectre v2 attacks, which is out-of-scope for this paper.} 
Consider the original access \codevar{y = A[x]}, say at line $l$, without Delta Pointers instrumentation, we have $\vioregionspec(l) = \emptyset$ as no mis-prediction occurs. However,  $\vioregionspec(\defenseseq{l}) = \{\cod{A[8]}, \cod{A[9]}, ...  \}$
since speculative execution may bypass the canonical address check at Line 4. Therefore, delta pointers is not safety-preserving, as $\vioregionspec(\defenseseq{l}) \not\subseteq \vioregionspec(l)$.

\textit{\safestack is safety-preserving.} \safestack~\cite{kuznetzov2018safestack} splits the program stack into two separate regions, where a safe stack stores safe objects whose accesses are all provably safe (i.e., $\vioregionseq(\location) = \emptyset$ for all $\location$ that may access the object), and an unsafe stack stores everything else. By design, it provides sequential (but not speculative) memory safety to those objects on the safe stack, and the protection is at the object granularity. 

For objects on the unsafe stack, however, \safestack enforces a coarse-grained software fault isolation (SFI) policy, ensuring that all accesses to unsafe objects remain within the bounds of the unsafe stack through the use of segment registers (e.g., \codevar{gs} and \codevar{fs}). The defense is dual safe for unsafe objects, but only at the granularity of the whole unsafe stack.

Moreover, because safe stack relies solely on base-and-offset arithmetic, it is safety-preserving as it does not introduce additional memory accesses under speculative execution.



\begin{table}[t]
\setlength{\tabcolsep}{3pt}
\centering
\small
\begin{tabular}{|c|c|c|c|}
\hline
Mechanism & Range & Safety & Preservation \\\hline
\hline
Baggy & Object & $\seq$ & No \\\hline
Delta Pointers & Object & $\seq$ & No \\\hline
\multirow{2}{*}{\safestack} & Safe: Object & $\seq$ & Yes \\
                             & Unsafe: Unsafe Stack$\star$ & Dual & Yes \\\hline
SLH & Object & $\mathghost$$\dagger$ & Yes \\\hline
SpecASan & Object & Dual & Yes \\\hline
CHERI & Object & Dual & Yes \\\hline
\hline
Baggy+SLH & Object & $\seq$ & No \\\hline
SLH+Baggy & Object & $\seq$ & No \\\hline
DeltaPointers+SLH & Object & $\seq$ & No \\\hline
SLH+DeltaPointers & Object & $\seq$ & No \\\hline
\multirow{2}{*}{SafeStack+SLH} & Safe: Object & Dual & Yes \\
                             & Unsafe: Unsafe Stack$\star$ & Dual & Yes \\\hline
\hline
\dualbaggy & Object & Dual & Yes \\\hline
\dualdp & Object & Dual & Yes \\\hline
\end{tabular}
\caption{Properties of Memory Safety Mechanisms. \\$\star$: SafeStack enforces a coarse-grained software-fault isolation policy on unsafe objects: accesses to unsafe objects are constrained within the unsafe stack. \\ 
$\dagger:$ SLH is speculative safe if the out-of-bound accesses stem from branch misprediction.}
\label{tab:property}
\end{table}

\textit{SLH is safety-preserving.} Recall the example in \Cref{sec:motivation}. SLH masks memory addresses that may become out-of-bounds under speculative execution using a dedicated register \codevar{mis}, which tracks the mis-speculation status. Accordingly, SLH is \emph{conditionally} sound when \codevar{mis} accurately reflects the status of the branch-prediction unit. However, bypassing hardware-enforced canonical address checks under speculative execution does not stem from branch misprediction. Consequently, applying SLH to such instructions does not yield sound speculative protection for these checks, instrumented by mechanisms like Baggy and Delta Pointers.

On the other hand, the masking operation does not change program semantics under sequential semantics, nor introducing new memory accesses. Hence, it ensures that $\vioregionseq(\defensespec{l}) \subseteq \vioregionseq(l)$. Therefore, SLH is safety-preserving. 

\textit{SpecASan is safety-preserving.} SpecASan~\cite{ganjisaffar2025specasan} extends the protection of ARM's hardware based defense, MTE, which enforces memory tag checks only under sequential execution, to also cover speculative execution. Hence, it is dual safe.  

SpecASan does not introduce new memory accesses during  enforcement.
So, SpecASan is safety-preserving.


\textit{CHERI is safety-preserving.} CHERI~\cite{woodruff2014cheri} replaces traditional integer pointers with \textit{capabilities} that include addresses, bounds, and permissions. Upon each memory access, bounds are directly extracted from the capability and compared against the address during both sequential and speculative executions. Hence, CHERI is dual safe. Moreover, CHERI does not add any new memory accesses. Thus, CHERI is safety-preserving.

\subsection{Dual safety properties.}
With \Cref{thm:composition}, it becomes straightforward to determine whether the composition of two mechanisms ensures dual safety. For example, applying Baggy or Delta Pointers, which are not safety preserving, on top of other mechanisms yields only sequential memory safety. One subtlety arises in compositions where SLH is applied on top of a sound sequential defense, since SLH is only conditionally sound. In particular, both Baggy and Delta Pointers rely on hardware enforcement through canonical address checks, for which SLH does not provide speculative protection. Consequently, applying SLH on top of Baggy or Delta Pointers does not guarantee dual memory safety.
%
In contrast, \safestack does not rely on Spectre-vulnerable hardware check, making ``\safestack+SLH'' a dual safe composition.

\subsection{Software-Based Dual Safe Mechanisms}
\label{sec:dualbaggy}

Based on the insights from dual safety analysis, we next present two \emph{novel} mechanisms that are dual memory safe.

\textit{\dualbaggy.}
%
%
Baggy allocates memory objects in fixed size slots, ensuring that each allocation size is a power of two. This design reduces out-of-bounds checks on pointer arithmetic to examining the high order bits of the pointer, which can be implemented via a bitwise XOR followed by a right shift.

Consider a 64 bit system with a 128 byte object. The lowest 5 bits encode the offset within a slot, while the highest 57 bits represent the object’s base address, uniquely identifying the allocation. The remaining 2 bits distinguish among the four slots allocated to this object.

\vspace{-3ex}
\begin{figure}[htp]
    \centering
    \includegraphics[width=0.6\linewidth]{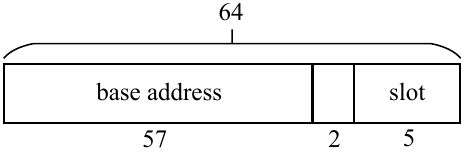}
    \label{fig:address}
\end{figure}
\vspace{-3ex}


Motivated by SLH, \dualbaggy removes the problematic branch-based and hardware-based checks in Baggy via a branchless \codefn{select}, which sets a computed address to -1 if the computed address goes out of bounds; otherwise, the pointer remains the same. Recall the example in \Cref{sec:motivation}, the transformed code by \dualbaggy is as follows. 


\begin{lstlisting}[
    linebackgroundcolor={%
    \ifnum\value{lstnumber}=2
            \color{blue!30}
    \fi
    \ifnum\value{lstnumber}=3
            \color{blue!30}
    \fi
    \ifnum\value{lstnumber}=4
            \color{blue!30}
    \fi
    }
]
addr = &A[x];
log_size = size_table[A >> log_slot_size];
is_oob = (addr XOR A) >> log_size != 0
final_addr = select(is_oob, -1, addr);
y = load(final_addr);
\end{lstlisting}

\textit{\dualdp.}
Delta Pointers leverage integer overflow to detect out of bounds accesses. However, as discussed in \Cref{sec:properties}, it relies on hardware checks that are vulnerable under speculative execution.

To fix this, we present \dualdp, a variant of Delta Pointers that explicitly inspects the most significant bit in software and sets the resulting address to -1 when this bit is set. For a memory access such as \codevar{y = A[x]} (where the size of \codevar{A} is 8), the instrumented code is as follows.

\begin{lstlisting}[
    linebackgroundcolor={%
    \ifnum\value{lstnumber}=2
            \color{blue!30}
    \fi
    \ifnum\value{lstnumber}=3
            \color{blue!30}
    \fi
    \ifnum\value{lstnumber}=4
            \color{blue!30}
    \fi
    \ifnum\value{lstnumber}=5
            \color{blue!30}
    \fi
    }
]
addr = &A[x];
new_addr = addr + (x << 32);
masked_addr = new_addr & ADDR_MASK;
is_oob = masked_addr < 0;
final_addr = select(is_oob, -1, masked_addr);
y = load(final_addr);
\end{lstlisting}

By performing a signed comparison with zero at Line 4, \dualdp detects if the MSB is set. If the pointer is OOB, the \codevar{select} instruction transforms the address to -1 (an inherently invalid address with all bits as 1) before the load occurs. This ensures that even during speculative execution, the address used for the memory access is safely neutralized, providing a robust protection that bridges the gap between software metadata and hardware enforcement.\footnote{Delta Pointers focuses on overflow checks since they are deemed more common than underflow~\cite{kroes2018delta}. Our dual safe variant follows the same direction. Recent work has shown that strategically combining different defenses, i.e., Baggy and Delta Pointers, could achieve overall spatial safety while keeping the overhead low~\cite{george2024optisan}, which we leave as future work.}

\subsection{Selective Defense}
The composition theorem enables systematic reasoning about dual memory safety. However, it conservatively assumes that memory safety mechanisms are applied to \emph{all} program locations with a nonempty violation set, which can lead to excessive overhead. We next present a stronger result that enables sound and more efficient composition when both defense mechanisms are sound and safety preserving. Without loss of generality, we consider the composition in which $\defensespec{}$ is applied \emph{after} $\defenseseq{}$; by symmetry, the same reasoning applies when the order is reversed.

Our insight is that when both defense mechanisms are sound and safety-preserving, speculative safety mechanisms can be omitted at program locations that are initially free from speculative memory errors. This is intuitively safe as the safety-preserving sequential defense cannot introduce new speculative memory errors. Consider the code snippet in \Cref{code:cat3}, where the out-of-bounds accesses at Lines 1 and 3 only occur under the sequential semantics. With the selective composition theorem below, we can avoid applying speculative defenses, such as SLH, to achieve dual memory safety.

More precisely, given sound and safety-preserving $\defensespec{}$ and $\defenseseq{}$, we define their selective composition as follows, where mechanisms are omitted on program locations that are safe:
\[\defenseselective{\location}=\begin{cases}
            \location, & \text{if $\vioregionspec(\location)=\emptyset=\vioregionseq(\location)$}\\
			\defenseseq{\location}, & \text{if $\vioregionspec(\location)=\emptyset$}\\
            \defensespec{\location}, & \text{if $\vioregionseq(\location)=\emptyset$}\\
            \defensespec{\defenseseq{\location}}, & \text{otherwise}
		 \end{cases}\]

\begin{theorem}[Selective Composition]
\label{thm:selective}
Let $\defenseselective{}$ be a selective speculative defense based on sound and safety-preserving mechanisms $\defenseseq{}$ and $\defensespec{}$.  $\defenseselective{c}$ is dual memory safe.
\end{theorem}
\begin{proof}
The first case is trivial and the last case directly comes from~\Cref{thm:composition}. For the second case, we only need to show that $\defenseseq{\location}$ is dual memory safe for program locations where $\vioregionspec(\location)=\emptyset$. Since $\defenseseq{}$ is sound, we know that $\vioregionseq(\defenseseq{\location})=\emptyset$. Moreover, we have $\vioregionspec(\defenseseq{\location})\subseteq \vioregionspec(\location)=\emptyset$, where the first subset relation holds as $\defenseseq{}$ is safety preserving. The third case is the dual of the second case.
\end{proof}

\subsection{Selective Dual Memory Safe Mechanisms}
\label{sec:seldefences}

To demonstrate the practical relevance of \Cref{thm:selective}, we consider three selective compositions that yield dual safety. In each composition, we first apply a dual safe mechanism to fix all sequential memory errors, i.e., program locations categorized as U$_{\seq}$S$_{\mathghost}$ and U$_{\seq}$U$_{\mathghost}$. Since the mechanism is dual safe, those locations are exempt from SLH protection. Then, we only apply SLH to program locations categorized as S$_{\seq}$U$_{\mathghost}$, which is safe due to \Cref{thm:selective}, as  $\vioregionseq(\location)=\emptyset$ for those locations. Therefore, no redundant protection is introduced in the following three compositions:

\begin{itemize}
    \item \baggyselslh: \dualbaggy is applied to program locations categorized as U$_{\seq}$S$_{\mathghost}$ and U$_{\seq}$U$_{\mathghost}$; SLH is applied to those categorized as S$_{\seq}$U$_{\mathghost}$.

    \item \dpselslh: in this configuration, \dualdp is applied to program locations categorized as U$_{\seq}$S$_{\mathghost}$ and U$_{\seq}$U$_{\mathghost}$; SLH is applied to those categorized as S$_{\seq}$U$_{\mathghost}$.

    \item \sfiselslh: 
    we put objects that are accessed by program locations categorized as U$_{\seq}$S$_{\mathghost}$ and U$_{\seq}$U$_{\mathghost}$ into the unsafe stack. Hence, SafeStack ensures dual memory safety for all accesses to the unsafe stack (\Cref{tab:property}). Moreover, SLH is applied to all program locations categorized as S$_{\seq}$U$_{\mathghost}$. 
    
\end{itemize}


One subtle point worth noting is that it may be tempting to place all unsafe objects (i.e., those accessed by program locations in U$_{\seq}$S$_{\mathghost}$, U$_{\seq}$U$_{\mathghost}$ and S$_{\seq}$U$_{\mathghost}$) onto the unsafe stack, and then apply \safestack as a dual memory safety mechanism for this stack. However, this approach does not preserve the original security guarantees of SafeStack, as illustrated by the following code fragment:
\begin{lstlisting}
x = A[0];
y = B[8];
if(i < 8)
    z = C[i];
\end{lstlisting}
Assume that arrays \codevar{A}, \codevar{B}, and \codevar{C} each have size 8. Under sequential execution, the accesses at Lines 1 and 4 are memory safe, while the access at Line 2 is unsafe. Accordingly, the original \safestack places \codevar{B} on the unsafe stack, and it places \codevar{A} and \codevar{C} on the safe stack. This placement ensures object level memory safety for arrays \codevar{A} and \codevar{C} (i.e., no out-of-bounds accesses can reach them), while any erroneous access to \codevar{B} remains confined to the unsafe stack. Our configuration of \sfiselslh extends the same guarantees to speculative execution.

Conservatively placing array \codevar{C} on the unsafe stack also ensures object level memory safety for \codevar{A}. However, this approach weakens the overall security guarantees, as \codevar{C} is no longer protected from accesses by other objects on the unsafe stack, given that they reside within the same sandbox.

\section{\analysis}
\label{sec:analysis}

We present \analysis, a framework that automatically enforces sound and efficient dual memory safety for target software, following the workflow shown in \Cref{fig:workflow}. \analysis begins with a two-phase analysis to identify vulnerable program locations. The first phase (\memphaseone) identifies locations that may exhibit sequential memory errors (i.e., those in U$_\seq$S$_{\mathghost}$ or U$_\seq$U$_{\mathghost}$), while the second phase (\memphasetwo) identifies locations that may exhibit either sequential or speculative memory errors (i.e., those in U$_\seq$S$_{\mathghost}$ or S$_\seq$U$_{\mathghost}$ or U$_\seq$U$_{\mathghost}$). Violations identified in both phases are then passed to a dual mitigation component, which applies dual memory safe mechanisms (\Cref{tab:property})  accordingly. 

\begin{figure}
    \centering
    \includegraphics[width=0.7\linewidth]{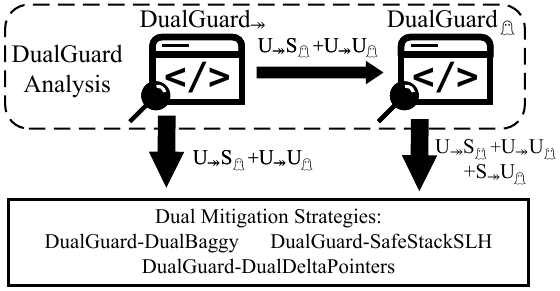}
    \vspace{-1ex}
    \caption{\framework Workflow}
    \label{fig:workflow}
\end{figure}

\subsection{Static Analysis}

Each analysis phase first collects all pointer arithmetic operations in the program, where potentially unsafe spatial accesses may occur, along with the sizes of the underlying memory objects being indexed. It then performs a forward data flow analysis to compute value ranges~\cite{huang2022taming, huang2024top} for numeric program variables, which are used to determine the possible offsets involved in pointer arithmetic operations.

In particular, we use a forward-may data-flow analysis to compute the range of each variable in the program, with the data-flow equation below:
\begin{equation}
     \begin{aligned}
         In(B) &= \bigcup_{B' \in pred(B)} Out(B') \\
         Out(B) &= (In(B) - kill(B)) \cup gen(B)
     \end{aligned}
    \label{equ:vr}
\end{equation}
Here, $In(B)$ is the value-range information at the entry point of basic block B, which is computed by taking the union of the range at the exit point of all its predecessors, i.e., $Out(B')$. 
When the range cannot be decided at compile time, the resulting range is conservatively estimated as unbounded.

$Out(B)$ denotes the range at the exit of basic block $B$. Its computation involves the sets $kill(B)$ and $gen(B)$, which capture the variables defined in $B$ together with their corresponding range information. By subtracting $kill(B)$ from $In(B)$ and then taking the union with $gen(B)$, the analysis removes outdated range information and incorporates the updated ranges introduced within $B$.

With all value-range information computed, the analysis determines whether an operation is safe at compile time by checking if the offset range lies within the bounds of the underlying object. Operations that violate this condition are marked as unsafe and reported accordingly.

Both \memphaseone and \memphasetwo follow the workflow described above, but differ in two aspects: (1) \memphaseone is path sensitive, and (2) \memphasetwo incorporates a speculative specific out-of-window analysis to exclude memory operations that are never executed speculatively.

\subsection{Path Sensitivity}

\memphaseone identifies memory errors under sequential semantics. Consequently, path conditions can be safely leveraged to restrict the range of variable values within conditional branches, enabling a path-sensitive analysis. In contrast, \memphasetwo operates under speculative-aware semantics, in which an instruction may be executed regardless of the corresponding path condition. As a result, the analysis must conservatively evaluate memory accesses without incorporating path conditions. This conservative treatment is intended to over-approximate all possible values that a variable may assume under speculative execution semantics.

\subsection{Out-of-Window Analysis}
\label{sec:oow}

\memphasetwo further incorporates an out-of-window analysis to eliminate instructions that cannot be executed speculatively. In modern processors, the number of speculatively executed instructions is bounded by the size of the reorder buffer. When the reorder buffer is full, speculative execution stalls until the corresponding branch condition is resolved. Consequently, each branch induces a finite speculation window within which instructions may execute speculatively. Instructions that fall outside the speculation windows of all branches are therefore never subject to speculative execution.

To identify these out-of-window instructions, \memphasetwo performs a reachability test from every branch instruction, up to the window size. Unreachable instructions are reported by \memphasetwo as speculatively memory safe.


\section{Evaluation}
\label{sec:evaluation}


\begin{figure*}[t]
    \centering
    \includegraphics[width=0.9\linewidth]{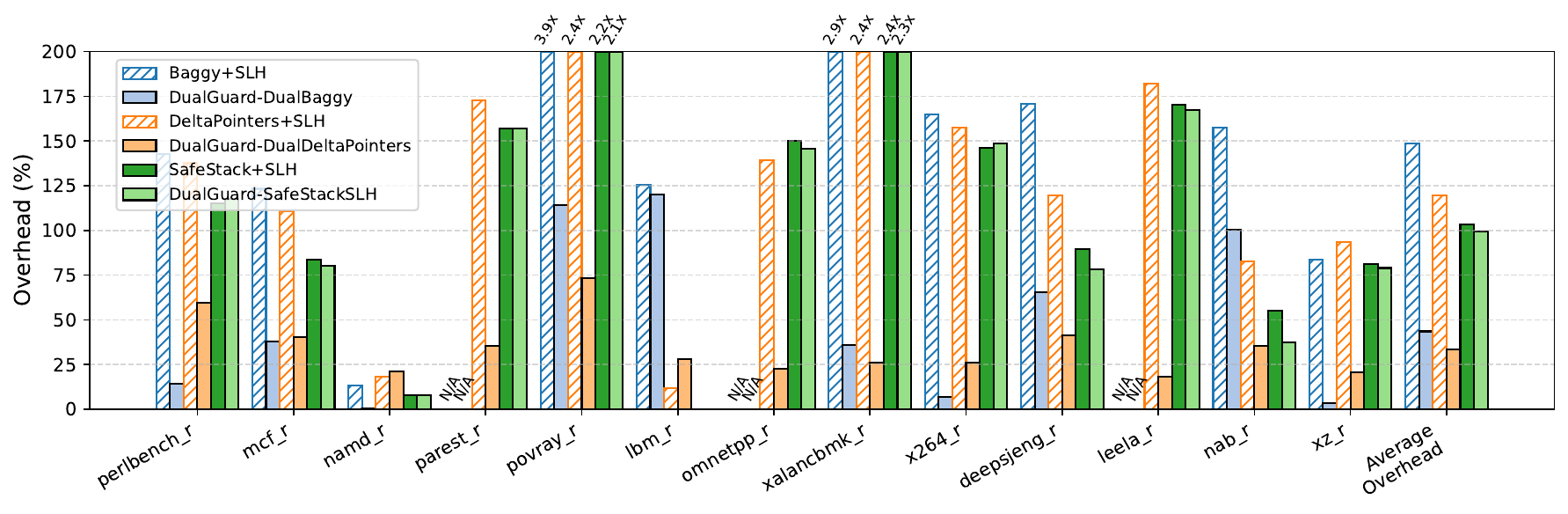}
    \vspace{-2ex}
    \caption{SPEC CPU 2017 overhead (Bars with the same color but different shades are directly comparable; bars with different colors represent different defense schemes and are not comparable with each other. Solid bars represent dual safe defenses, while bars with stripes are not dual safe.)}
    \label{fig:spec-overhead}
\end{figure*}

\subsection{Implementation}


We implement \analysis by extending the memory safety analysis of~\cite{huang2024comprehensive}, building on its LLVM-based infrastructure with approximately 1,500 additional lines of C++ code. Our extension also refines the value ranges of path conditions, improving the precision of sequential reasoning.
To leverage LLVM’s existing SLH infrastructure, which operates closer to the assembly level, \analysis inserts dummy instructions at vulnerable program locations identified by \memphasetwo, allowing them to be mitigated in later compilation stages.
The out-of-window analysis (\Cref{sec:oow}) is implemented in approximately 700 lines of code. Consistent with prior work~\cite{guarnieri2020spectector}, we set the window size to 200 instructions, providing a conservative estimate of the reorder buffer capacity in contemporary processors. \changes{We note that the window size is fully configurable in \framework: one can easily tune it to match the re-order buffer size of their target hardware.}

For both \safestack and SLH, we modify their implementations within the LLVM infrastructure so that they are applied to the potentially unsafe objects identified by our static analysis, rather than those selected by LLVM’s default heuristics. For Baggy, we build on an open source implementation from~\cite{george2024optisan} and implement a dual memory safe variant, \dualbaggy, as described in \Cref{sec:dualbaggy}.

For Delta Pointers, we build on its public repository~\cite{deltapointerrepo} and introduce two additional LLVM IR instructions, \codefn{icmp} and \codefn{select}. In later compilation stages, \codefn{select} is lowered to \codefn{cmov}, which uses the \codevar{EFLAGS} register to determine whether to perform a \codefn{mov} operation. The \codefn{icmp} instruction is optimized away, as the relevant \codevar{EFLAGS} are already set by the \codefn{and} instruction introduced by the original Delta Pointers mechanism. Hence, at the assembly level, only a branchless \codefn{cmov} is added, allowing \dualdp to remain efficient while providing stronger security guarantees.



\subsection{Evaluation Setup}

We evaluate \framework on two benchmark suites. The first suite is the widely used SPEC CPU 2017 benchmark suite. \changes{Of the 16 benchmarks fully implemented in C/C++, which can be analyzed by our \framework implementation, we successfully evaluated \framework on 13: perlbench\_r, mcf\_r, namd\_r, parest\_r, povray\_r, lbm\_r, omnetpp\_r, xalancbmk\_r, x264\_r, deepsjeng\_r, leela\_r, nab\_r, and xz\_r. Of the three excluded benchmarks, imagick\_r crashes under both Baggy and Delta Pointers because it intentionally performs out-of-bounds memory accesses; both tools are designed to abort execution upon detecting such spatial memory violations. We exclude gcc\_r and blender\_r because our context-sensitive whole-program analysis either times out or exhausts available memory. We leave improving the scalability of our prototype as future work.}

The second set of benchmarks is a collection of cryptographic libraries used in a prior study on SLH-based defenses~\cite{zhu2025lightslh}. The evaluation is performed on a server (Ubuntu 24.04.3 LTS) with an Intel(R) Xeon(R) 6517P CPU (3.2 GHz) and 256GB RAM.

We aim to answer the following research questions:
\begin{enumerate}
    \item[\textbf{Q1:}] What is the performance overhead among \framework's dual safe mechanisms? 
    
    \item[\textbf{Q2:}] How does the overhead of \framework compare to naive stacking of sequential and speculative defenses?
    
    \item[\textbf{Q3:}] Is the overhead of \framework comparable with state-of-the-art defenses that only provide either sequential or speculative safety? 
    
    \item[\textbf{Q4:}] How much do selective defense mechanisms (\Cref{sec:seldefences}) help to reduce overhead?
    
    \item[\textbf{Q5:}] Does \framework scale to real-world programs?
\end{enumerate}

\subsection{The Overhead of \framework}
\label{sec:eva-dual}

\framework provides three dual-safe configurations:
\begin{itemize}
    \item \frameworkbaggy applies dual Baggy to all potentially unsafe instructions.

    \item \frameworkdp applies dual Delta Pointers to all potentially unsafe instructions.

    \item \frameworksfislh is the selective composition of Safe Stack and SLH (aka. \sfiselslh in \Cref{sec:seldefences}).
\end{itemize}

\changes{We did not include CHERI~\cite{woodruff2014cheri} and SpecASan~\cite{ganjisaffar2025specasan} in our evaluation because both require specialized hardware support, whereas \dualbaggy and \dualdp are software-only defenses. Specifically, CHERI depends on dedicated architectural extensions, while SpecASan relies on the ARM Memory Tagging Extension (MTE) and does not have a publicly available implementation.}

For the other two selective composition variants \baggyselslh and \dpselslh introduced in \Cref{sec:seldefences}, we study their performance overhead under Q4 (\Cref{sec:eval-selective}).

To address Q1 and Q2, we assess the performance overhead of the three configurations of \framework and compare them against their corresponding naive compositions, which sequentially apply three mechanisms to all potentially unsafe program locations \changes{identified by our static analysis}: \baselinebaggy, \baselinedp, and \baselinesfi. It is important to note that \baselinebaggy and \baselinedp do not satisfy dual safety (\Cref{tab:property}); nevertheless, we include them to emphasize that, despite their lack of safety guarantees, these naive compositions still incur substantial performance overhead.

We measure the overhead of each configuration against the original program execution time without any protection. 

\begin{figure*}[t]
    \centering
    \includegraphics[width=0.95\linewidth]{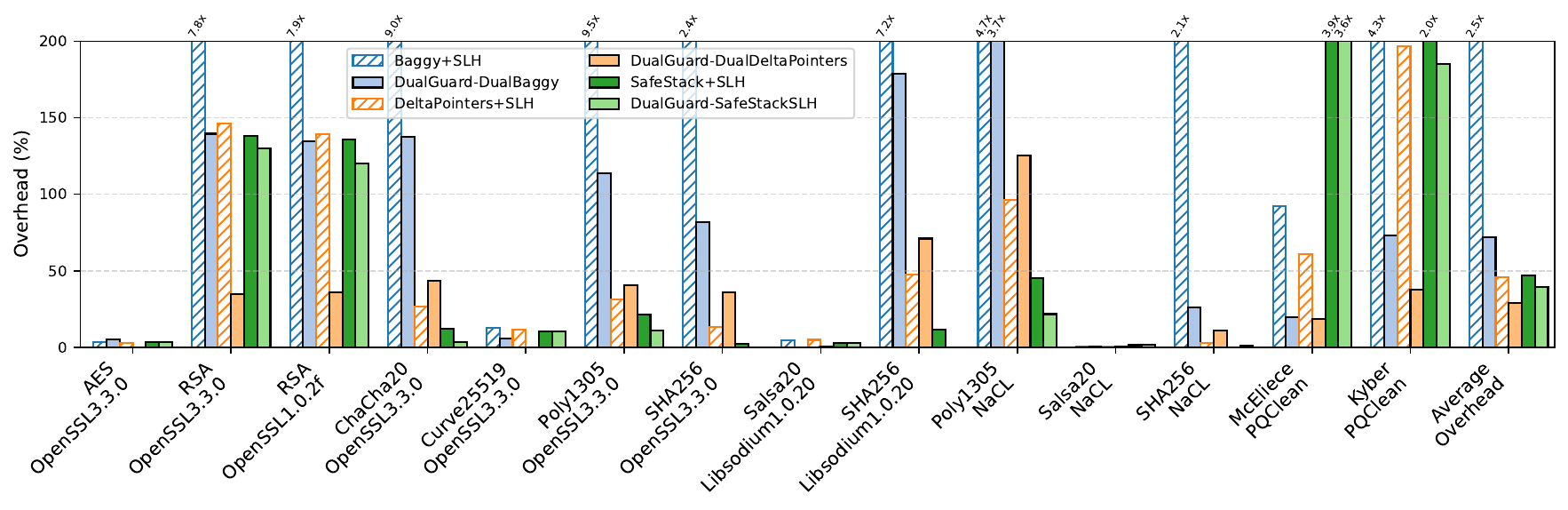}
    \vspace{-2ex}
    \caption{Cryptographic library overhead (Bars with the same color but different shades are directly comparable; bars with different colors represent different defense schemes and are not comparable with each other. Solid bars represent dual safe defenses, while bars with stripes are not dual safe.)}
    \label{fig:crypto-overhead}
\end{figure*}

\subsubsection{SPEC CPU 2017} We summarize the results in \Cref{fig:spec-overhead}, where the striped bars represent ones that are not dual safe.

\textbf{Baggy-based protection.} 
The overhead of \baselinebaggy ranges from 7\% to 358.5\%. In contrast, \frameworkbaggy consistently incurs lower overhead. The reduction in overhead ranges from 3\% for lbm\_r to 95.3\% for x264\_r. On average,\footnote{All averages are reported as geometric means.} \baselinebaggy incurs an overhead of 148.4\%, whereas \frameworkbaggy reduces it to 43.6\%. Overall, \framework achieves an overhead reduction of 70.6\% over the ``mitigation-upon-mitigation'' option based on Baggy.

\textbf{Delta Pointers-based protection.} The overhead of \baselinedp ranges from 11.6\% to 241.9\%. In contrast, \frameworkdp consistently incurs lower overhead except for 2 benchmarks where it incurs a slightly higher overhead. On average, \baselinedp incurs an overhead of 119.7\%, whereas \frameworkdp reduces it to 33.6\%. Overall, \framework achieves an overhead reduction of 71.9\% over \baselinedp.

\textbf{{\safestack}-based protection.} 
The overhead of \baselinesfi ranges from 0\% to 240.1\%. \frameworksfislh incurs lower or comparable overhead than \baselinesfi in all benchmarks.
Overall, \baselinesfi incurs an overhead of 103.2\% whereas \frameworksfislh reduces it to 99.5\%, though we observe that the reduction is smaller than what we observe in the Baggy-based and Delta-Pointers-based cases.

Overall, \framework’s dual-safe configurations, particularly the novel mechanisms \dualbaggy and \dualdp, consistently outperform their corresponding ``mitigation-upon-mitigation'' counterparts on the SPEC CPU 2017 benchmark. Among them, \dualdp incurs the lowest overhead overall.

\subsubsection{Cryptographic libraries} Similar to the results on SPEC CPU 2017 benchmark, \framework’s dual-safe configurations consistently outperform their corresponding ``mitigation-upon-mitigation'' counterparts on the cryptographic libraries as well. The results are summarized in \Cref{fig:crypto-overhead}.

Across Baggy-, Delta Pointers-, and SafeStack-based protections, the overhead of these naive combinations ranges from 0\% to 945.1\%. In contrast, \framework’s configurations consistently incur lower overhead. The reduction in overhead is up to 97.6\% for Baggy, 95.4\% for Delta Pointers and 73.5\% for SafeStack. On average, \frameworkbaggy, \frameworkdp, and \frameworksfislh incur overheads of 71.9\%, 28.9\%, and 39.2\%, respectively. Overall, we observe reduction rates ranging from 16.8\% to 71.8\% across all configurations. Similar to the SPEC CPU 2017 benchmark, \dualdp achieves the lowest overall overhead of 28.9\%.

To better understand how \framework’s configurations achieve lower overhead than naive compositions, we also present an in-depth case study of a real-world example from the big number (bn) library of OpenSSL 3.3.0 in the Appendix. \changes{Overall, we observe that the overhead reduction mostly comes from: 1) replacing branch-based checks with more efficient branchless checks, and 2) eliminating redundant speculative instrumentation (e.g., SLH) enabled by our dual-safe defenses. We further observe that varying the reorder buffer size configuration has minimal impact on performance.}

\subsection{Dual vs. Partial Memory Safety}

To address Q3, we compare \frameworkbaggy and \frameworkdp with the SLH implementation in LLVM (LLVM SLH), \lightslh~\cite{zhu2025lightslh}, Baggy, and Delta Pointers, as all of these mechanisms enforce memory safety at the granularity of objects. 
Among these approaches, only the \framework variants offer dual memory safety guarantees. In contrast, Baggy and Delta Pointers fail to provide speculative memory safety, while LLVM SLH and \lightslh are limited to mitigating Spectre V1 attacks. 
\Cref{tab:partial} presents a summary of the overall performance overhead incurred by these mechanisms on the \changes{SPEC and} cryptographic library benchmarks; a detailed breakdown is included in the Appendix.

\begin{table}[t]
\caption{Compare \framework with partial protections \\ (*: LightSLH cannot run on SPEC benchmarks due to insufficient instruction support and scalability issues.)}
\label{tab:partial}
\centering
\begin{tabular}{l|c|r|r} 
\toprule
\multirow{2}{*}{Defense}    & \multirow{2}{*}{Safety} & \multicolumn{2}{c}{Overhead}  \\ 
\cline{3-4}
                            &                         & SPEC &   Crypto              \\ 
\midrule
\midrule
LLVM-SLH                    & $\mathghost$   &  116\%  &  46\%               \\
LighSLH                     & $\mathghost$   &  N/A*  &  31\%                 \\
Baggy                       & $\seq$         &  131\% &    147\%             \\
Delta Pointers              & $\seq$         &  24\%  &  20\%                \\
\dualbaggy                  & Dual           &  44\%   &  72\%               \\
\dualdp                     & Dual           &  34\%  &     29\%             \\
\bottomrule
\end{tabular}
\end{table}
We observe that the \framework defenses deliver dual memory safety guarantees while incurring only modest overhead relative to Delta Pointers and, in some cases, even lower overhead than Baggy, LLVM-SLH, and LightSLH, all while maintaining or strengthening memory safety assurances across both sequential and speculative execution contexts.


\subsection{Selective Defense}
\label{sec:eval-selective}

\begin{figure}[h]
    \centering
    \includegraphics[width=0.7\linewidth]{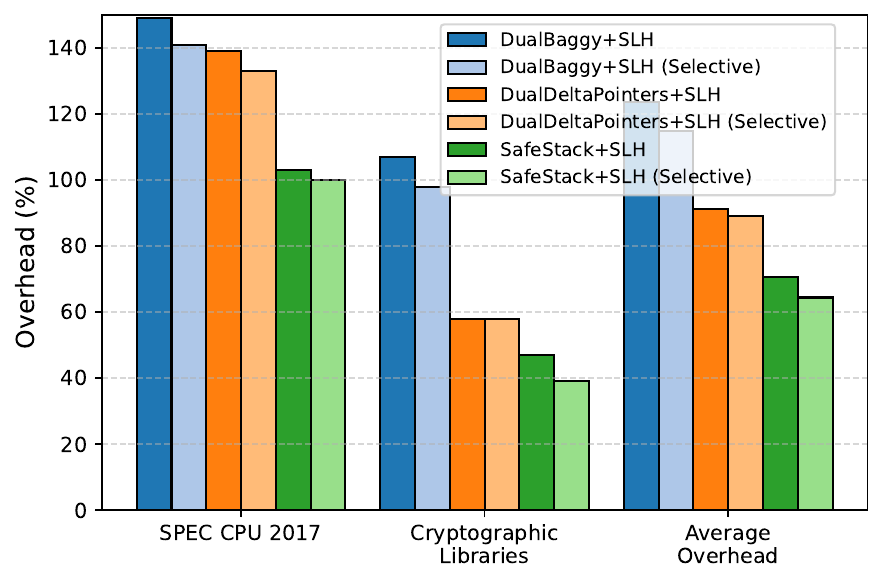}
    \vspace{-2ex}
    \caption{The effect of selective protection}
    \label{fig:selective}
\end{figure}

To address Q4, we compare the three selective dual memory safety mechanisms described in \Cref{sec:seldefences} against their counterparts that do not leverage~\Cref{thm:selective}, i.e., approaches that conservatively apply protections in sequence to all program locations classified as U$_{\seq}$S$_{\mathghost}$, U$_{\seq}$U$_{\mathghost}$ or U$_{\seq}$S$_{\mathghost}$. The results are summarized in~\Cref{fig:selective}, with a detailed per-benchmark breakdown provided in the Appendix.

We observe that, across Baggy-, Delta Pointers-, and SafeStack-based protections, selective composition reduces performance overhead by 5.4\%, 4.3\%, and 2.9\%, respectively, on the SPEC CPU 2017 benchmark, and by 8.4\%, 0\%, and 17\%, respectively, on cryptographic library benchmark. This reduction underscores the practical value of \Cref{thm:selective}.

\subsection{Scalability}

The cryptographic library benchmarks are relatively small, with most programs comprising fewer than 1 KLoC; \framework’s static analysis completes for all of them within 8 seconds. For the SPEC CPU benchmarks, the largest program is xalancbmk\_r (520 KLoC), which requires 51.23 seconds for analysis, while perlbench\_r requires the longest analysis time of 8254.01 seconds. These results indicate that \framework can analyze real-world codebases of moderate size within a reasonable timeframe. A detailed breakdown of analysis times for each benchmark is provided in the Appendix.


\section{Limitations and Future Work}
\label{sec:future_safety}

Within the scope of spatial memory safety under Spectre~v1 attacks, our work provides the first systematic investigation into the interaction between memory safety and speculative execution, alongside proposing concrete spatial dual memory safe defense mechanisms. Beyond that, we believe that the formal insights and methodology developed in this paper establish a foundation for exploring broader threat models. For example, our safety preservation abstraction is formulated over general program transformations rather than being inherently bound to spatial bounds checking or Spectre~v1 mitigations. Consequently, it can be extended to evaluate whether other defenses (e.g., retpolines~\cite{kadir2019retpoline}, IBRS~\cite{intel2018ibrs}, or reference counting~\cite{reinking2021perceus}) preserve sequential and speculative safety guarantees. 

Moreover, there are at least two future directions for extending \framework:

\begin{itemize}
    \item \textbf{Temporal Memory Safety:} Temporal defenses (e.g., reference counting~\cite{reinking2021perceus}, and garbage collection~\cite{wilson1992uniprocessor}) maintain object lifetime metadata rather than spatial bounds. Under speculative execution, stale or mispredicted metadata could lead to transient use-after-free conditions or premature object reclamation. 
    
    \item \textbf{Other speculative attacks:} Other Spectre variants, such as Specter~v2 (indirect branch target poisoning)~\cite{barberis2022branch} and Specter-RSB (return stack buffer manipulation)~\cite{koruyeh2018spectre}, hijack transient control flow to arbitrary targets. Such indirect mispredictions could speculatively bypass memory-safety checks entirely or jump past instrumentation boundaries. 
\end{itemize}

These attack surfaces introduce failure modes that are qualitatively distinct from the bounds-check bypass scenarios addressed in this paper. Adapting our framework and verification models to account for dynamic control-flow redirection and temporal lifecycle invariants presents a promising direction for future research.

\section{Conclusion}

This paper introduces \framework, a unified framework to systematically address the fragmentation between non-speculative and speculative memory safety. By formalizing memory access behaviors into four distinct categories, \framework can identify the minimal requirements for achieving dual protection. Our evaluation confirms that this unified approach successfully removes redundant protections that often lead to unnecessary performance penalties. \framework achieves dual memory safety with a modest overhead, offering a robust and efficient solution for modern systems vulnerable to both sequential and speculative memory exploitations. Future work will continue to explore the scalability of the framework and generalize it to other Spectre attacks.

\section*{Ethics considerations}
We have meticulously considered the ethical implications of developing \framework, a framework designed to systematically provide memory safety across sequential and speculative execution. Our goal is to mitigate existing vulnerabilities. We do not reveal any new attack surface or specific attack sequences. To ensure our work was conducted responsibly, we adhered to a rigorous ethical framework throughout the research process. We identified three primary stakeholder groups potentially impacted by our work: End-Users of the affected software systems, Software Vendors and Maintainers, and Malicious Actors. 

\textbf{End-Users.} The individuals and organizations relying on safe software systems (e.g., cryptographic libraries) free from spatial memory errors are the ultimate beneficiaries of this research, but also the group most vulnerable to potential exploitation. To ensure no negative outcomes for end-users during the research phase, all experimental validations were conducted strictly within a private, isolated laboratory environment. Consequently, no legitimate production systems were exploited or compromised during our evaluation.


\textbf{Software Vendors and Maintainers}. To mitigate the risk of these vulnerabilities being exploited before defenses are in place, we adhered to standard responsible disclosure practices. As a defensive framework, we aim to ensure security by identifying and preventing potential vulnerabilities. Specifically, we identify unreported half-gadgets as potentially exploitable. However, these half-gadgets do not consist of complete exploits. Rather, they pointed out places where proactive defenses can be useful.



\textbf{Malicious Actors (Attackers).} A potential risk is that malicious actors might attempt to exploit program locations analyzed during our evaluation. We mitigate this risk in multiple ways:
\begin{itemize}
    \item Defensive Focus: Our primary objective is to develop secure and efficient defenses rather than uncover new attack surfaces. In particular, our case studies and analyses focus strictly on demonstrating how \framework{} eliminates redundant mitigations and reduces performance overhead, without presenting end-to-end exploits or actionable attack primitives.
    \item Guidance for Defenders: Rather than facilitating exploitation, we provide concrete principles and mechanisms that assist defenders in evaluating their systems and deploying targeted, dual-safe protections against realistic threats.
\end{itemize}


\section*{Acknowledgment}
This research is supported by the National Science Foundation under awards No. 2348304, No. 2401496 and No. 2401182, the Office of Naval Research under Award No. N000142512413, as well as the Air Force Research Laboratory (AFRL) under Contract No. FA8750-25-C-B041.

\bibliographystyle{IEEEtran}
\bibliography{ref}

@article{Oleksenko:2018,
  author    = {Oleksii Oleksenko and
               Bohdan Trach and
               Tobias Reiher and
               Mark Silberstein and
               Christof Fetzer},
  title     = {You Shall Not Bypass: Employing data dependencies to prevent Bounds
               Check Bypass},
  journal   = {CoRR},
  volume    = {abs/1805.08506},
  year      = {2018},
  url       = {http://arxiv.org/abs/1805.08506},
  archivePrefix = {arXiv},
  eprint    = {1805.08506},
  bibsource = {dblp computer science bibliography, https://dblp.org}
}

@misc{slh,
title={Speculative Load Hardening}, 
url={https://llvm.org/docs/SpeculativeLoadHardening.html}, 
journal={LLVM Documentation}, 
publisher={LLVM}, 
author={Chandler Carruth}
}

@inproceedings{cowan2000buffer,
  title={Buffer overflows: Attacks and defenses for the vulnerability of the decade},
  author={Cowan, Crispin and Wagle, F and Pu, Calton and Beattie, Steve and Walpole, Jonathan},
  booktitle={Proceedings DARPA Information Survivability Conference and Exposition. DISCEX'00},
  volume={2},
  pages={119--129},
  year={2000},
  organization={IEEE}
}

@incollection{kuznetzov2018safestack,
  title={Code-pointer integrity},
  author={Kuznetzov, Volodymyr and Szekeres, L{\'a}szl{\'o} and Payer, Mathias and Candea, George and Sekar, R and Song, Dawn},
  booktitle={The Continuing Arms Race: Code-Reuse Attacks and Defenses},
  pages={81--116},
  year={2018}
}

@inproceedings{akritidis2009baggy,
  title={Baggy Bounds Checking: An Efficient and Backwards-Compatible Defense against Out-of-Bounds Errors.},
  author={Akritidis, Periklis and Costa, Manuel and Castro, Miguel and Hand, Steven},
  booktitle={USENIX Security Symposium},
  volume={10},
  pages={96},
  year={2009}
}

@inproceedings{jones1997backwards,
  title={Backwards-Compatible Bounds Checking for Arrays and Pointers in C Programs.},
  author={Jones, Richard WM and Kelly, Paul HJ},
  booktitle={AADEBUG},
  volume={97},
  pages={13--26},
  year={1997}
}

@inproceedings{dhurjati2006backwards,
  title={Backwards-compatible array bounds checking for C with very low overhead},
  author={Dhurjati, Dinakar and Adve, Vikram},
  booktitle={Proceedings of the 28th international conference on Software engineering},
  pages={162--171},
  year={2006}
}

@article{saileshwar2022heapcheck,
  title={Heapcheck: Low-cost hardware support for memory safety},
  author={Saileshwar, Gururaj and Boivie, Rick and Chen, Tong and Segal, Benjamin and Buyuktosunoglu, Alper},
  journal={ACM Transactions on Architecture and Code Optimization (TACO)},
  volume={19},
  number={1},
  pages={1--24},
  year={2022},
  publisher={ACM New York, NY}
}

@article{nagarakatte2012watchdog,
  title={Watchdog: Hardware for safe and secure manual memory management and full memory safety},
  author={Nagarakatte, Santosh and Martin, Milo MK and Zdancewic, Steve},
  journal={ACM SIGARCH Computer Architecture News},
  volume={40},
  number={3},
  pages={189--200},
  year={2012},
  publisher={ACM New York, NY, USA}
}

@article{nagarakatte2013hardware,
  title={Hardware-enforced comprehensive memory safety},
  author={Nagarakatte, Santosh and Martin, Milo MK and Zdancewic, Steve},
  journal={IEEE Micro},
  volume={33},
  number={3},
  pages={38--47},
  year={2013},
  publisher={IEEE}
}

@article{necula2005ccured,
  title={CCured: Type-safe retrofitting of legacy software},
  author={Necula, George C and Condit, Jeremy and Harren, Matthew and McPeak, Scott and Weimer, Westley},
  journal={ACM Transactions on Programming Languages and Systems (TOPLAS)},
  volume={27},
  number={3},
  pages={477--526},
  year={2005},
  publisher={ACM New York, NY, USA}
}

@inproceedings{jim2002cyclone,
  title={Cyclone: a safe dialect of C.},
  author={Jim, Trevor and Morrisett, J Gregory and Grossman, Dan and Hicks, Michael W and Cheney, James and Wang, Yanling},
  booktitle={USENIX Annual Technical Conference, General Track},
  pages={275--288},
  year={2002}
}

@inproceedings{durumeric2014matter,
  title={The matter of heartbleed},
  author={Durumeric, Zakir and Li, Frank and Kasten, James and Amann, Johanna and Beekman, Jethro and Payer, Mathias and Weaver, Nicolas and Adrian, David and Paxson, Vern and Bailey, Michael and others},
  booktitle={Proceedings of the 2014 conference on internet measurement conference},
  pages={475--488},
  year={2014}
}

@article{orman2003morris,
  title={The Morris worm: A fifteen-year perspective},
  author={Orman, Hilarie},
  journal={IEEE Security \& Privacy},
  volume={1},
  number={5},
  pages={35--43},
  year={2003},
  publisher={IEEE}
}

@article{simpson2013memsafe,
  title={MemSafe: ensuring the spatial and temporal memory safety of C at runtime},
  author={Simpson, Matthew S and Barua, Rajeev K},
  journal={Software: Practice and Experience},
  volume={43},
  number={1},
  pages={93--128},
  year={2013},
  publisher={Wiley Online Library}
}

@article{huang2022taming,
  title={The taming of the stack: Isolating stack data from memory errors},
  author={Huang, Kaiming and Huang, Yongzhe and Payer, Mathias and Qian, Zhiyun and Sampson, Jack and Tan, Gang and Jaeger, Trent},
  year={2022},
  publisher={Internet Society}
}

@inproceedings{huang2024top,
  title={Top of the heap: Efficient memory error protection of safe heap objects},
  author={Huang, Kaiming and Payer, Mathias and Qian, Zhiyun and Sampson, Jack and Tan, Gang and Jaeger, Trent},
  booktitle={Proceedings of the 2024 on ACM SIGSAC Conference on Computer and Communications Security},
  pages={1330--1344},
  year={2024}
}

@article{huang2024comprehensive,
  title={Comprehensive memory safety validation: An alternative approach to memory safety},
  author={Huang, Kaiming and Payer, Mathias and Qian, Zhiyun and Sampson, Jack and Tan, Gang and Jaeger, Trent},
  journal={IEEE Security \& Privacy},
  volume={22},
  number={4},
  pages={40--49},
  year={2024},
  publisher={IEEE}
}

@inproceedings{Kocher2018spectre,
 author = {Paul Kocher and Jann Horn and Anders Fogh and and Daniel Genkin and Daniel Gruss and Werner Haas and Mike Hamburg and Moritz Lipp and Stefan Mangard and Thomas Prescher and Michael Schwarz and Yuval Yarom},
 title = {Spectre Attacks: Exploiting Speculative Execution},
 booktitle = {40th IEEE Symposium on Security and Privacy (S\&P'19)},
 year = {2019},
}

@inproceedings{Lipp2018meltdown,
 author = {Moritz Lipp and Michael Schwarz and Daniel Gruss and Thomas Prescher and Werner Haas and Anders Fogh and Jann Horn and Stefan Mangard and Paul Kocher and Daniel Genkin and Yuval Yarom and Mike Hamburg},
 title = {Meltdown: Reading Kernel Memory from User Space},
 booktitle = {27th {USENIX} Security Symposium ({USENIX} Security 18)},
 year = {2018},
}

@misc{intel2018lfence,
  title={Speculative Execution Side Channel Mitigations},
  author={{Intel Corporation}},
  url={https://www.intel.com/content/dam/develop/external/us/en/documents/336996-speculative-execution-side-channel-mitigations.pdf},
  year={2018}
}

@inproceedings{cauligi2022sok,
  title={SoK: Practical foundations for software Spectre defenses},
  author={Cauligi, Sunjay and Disselkoen, Craig and Moghimi, Daniel and Barthe, Gilles and Stefan, Deian},
  booktitle={2022 IEEE Symposium on Security and Privacy (SP)},
  pages={666--680},
  year={2022},
  organization={IEEE}
}

@inproceedings{canella2019systematic,
  title={A systematic evaluation of transient execution attacks and defenses},
  author={Canella, Claudio and Van Bulck, Jo and Schwarz, Michael and Lipp, Moritz and Von Berg, Benjamin and Ortner, Philipp and Piessens, Frank and Evtyushkin, Dmitry and Gruss, Daniel},
  booktitle={28th USENIX Security Symposium (USENIX Security 19)},
  pages={249--266},
  year={2019}
}

@inproceedings{zhu2025lightslh,
  title={Place protections at the right place: targeted hardening for cryptographic code against spectre v1},
  author={Zhu, Yiming and Huang, Wenchao and Xiong, Yan},
  booktitle={34th USENIX Security Symposium (USENIX Security 25)},
  pages={2557--2576},
  year={2025}
}

@inproceedings{baumann2025fslh,
  title={FSLH: Flexible Mechanized Speculative Load Hardening},
  author={Baumann, Jonathan and Blanco, Roberto and Ducruet, L{\'e}on and Harwig, Sebastian and Hri{\c{t}}cu, C{\u{a}}t{\u{a}}lin},
  booktitle={2025 IEEE 38th Computer Security Foundations Symposium (CSF)},
  pages={569--584},
  year={2025},
  organization={IEEE}
}

@inproceedings{shivakumar2023selslh,
  title={Typing high-speed cryptography against spectre v1},
  author={Shivakumar, Basavesh Ammanaghatta and Barthe, Gilles and Gr{\'e}goire, Benjamin and Laporte, Vincent and Oliveira, Tiago and Priya, Swarn and Schwabe, Peter and Tabary-Maujean, Lucas},
  booktitle={2023 IEEE Symposium on Security and Privacy (SP)},
  pages={1094--1111},
  year={2023},
  organization={IEEE}
}

@inproceedings{di2025half,
  title={Half Spectre, Full Exploit: Hardening Rowhammer Attacks with Half-Spectre Gadgets},
  author={Di Dio, Andrea and Hertogh, Math{\'e} and Giuffrida, Cristiano},
  booktitle={2025 IEEE Symposium on Security and Privacy (SP)},
  pages={3583--3598},
  year={2025},
  organization={IEEE}
}

@inproceedings{christou2024eclipse,
  title={Eclipse: Preventing Speculative Memory-error Abuse with Artificial Data Dependencies},
  author={Christou, Neophytos and Gaidis, Alexander J and Atlidakis, Vaggelis and Kemerlis, Vasileios P},
  booktitle={Proceedings of the 2024 on ACM SIGSAC Conference on Computer and Communications Security},
  pages={3913--3927},
  year={2024}
}

@inproceedings{mambretti2021bypassing,
  title={Bypassing memory safety mechanisms through speculative control flow hijacks},
  author={Mambretti, Andrea and Sandulescu, Alexandra and Sorniotti, Alessandro and Robertson, William and Kirda, Engin and Kurmus, Anil},
  booktitle={2021 IEEE European Symposium on Security and Privacy (EuroS\&P)},
  pages={633--649},
  year={2021},
  organization={IEEE}
}

@inproceedings{bhattacharyya2020specrop,
  title={$\{$SpecROP$\}$: Speculative exploitation of $\{$ROP$\}$ chains},
  author={Bhattacharyya, Atri and S{\'a}nchez, Andr{\'e}s and Koruyeh, Esmaeil M and Abu-Ghazaleh, Nael and Song, Chengyu and Payer, Mathias},
  booktitle={23rd International Symposium on Research in Attacks, Intrusions and Defenses (RAID 2020)},
  pages={1--16},
  year={2020}
}

@inproceedings{ravichandran2022pacman,
  title={PACMAN: attacking ARM pointer authentication with speculative execution},
  author={Ravichandran, Joseph and Na, Weon Taek and Lang, Jay and Yan, Mengjia},
  booktitle={Proceedings of the 49th Annual International Symposium on Computer Architecture},
  pages={685--698},
  year={2022}
}

@inproceedings{goktas2020speculative,
  title={Speculative probing: Hacking blind in the Spectre era},
  author={G{\"o}ktas, Enes and Razavi, Kaveh and Portokalidis, Georgios and Bos, Herbert and Giuffrida, Cristiano},
  booktitle={Proceedings of the 2020 ACM SIGSAC Conference on Computer and Communications Security},
  pages={1871--1885},
  year={2020}
}

@misc{intel2018ibrs,
  title={Indirect Branch Restricted Speculation},
  author={{Intel Corporation}},
  url={https://www.intel.com/content/www/us/en/developer/articles/technical/software-security-guidance/technical-documentation/indirect-branch-restricted-speculation.html},
  year={2018}
}

@misc{arm2025csdb,
  title={Consumption of Speculative Data Barrier},
  author={ARM},
  url={https://developer.arm.com/documentation/ddi0602/2022-09/Base-Instructions/CSDB--Consumption-of-Speculative-Data-Barrier-},
  year={2025}
}

@article{kim2014flipping,
  title={Flipping bits in memory without accessing them: An experimental study of DRAM disturbance errors},
  author={Kim, Yoongu and Daly, Ross and Kim, Jeremie and Fallin, Chris and Lee, Ji Hye and Lee, Donghyuk and Wilkerson, Chris and Lai, Konrad and Mutlu, Onur},
  journal={ACM SIGARCH Computer Architecture News},
  volume={42},
  number={3},
  pages={361--372},
  year={2014},
  publisher={ACM New York, NY, USA}
}

@inproceedings{zhang2023ultimate,
  title={Ultimate $\{$SLH$\}$: Taking speculative load hardening to the next level},
  author={Zhang, Zhiyuan and Barthe, Gilles and Chuengsatiansup, Chitchanok and Schwabe, Peter and Yarom, Yuval},
  booktitle={32nd USENIX Security Symposium (USENIX Security 23)},
  pages={7125--7142},
  year={2023}
}

@inproceedings{guarnieri2020spectector,
  title={Spectector: Principled detection of speculative information flows},
  author={Guarnieri, Marco and K{\"o}pf, Boris and Morales, Jos{\'e} F and Reineke, Jan and S{\'a}nchez, Andr{\'e}s},
  booktitle={2020 IEEE Symposium on Security and Privacy (SP)},
  pages={1--19},
  year={2020},
  organization={IEEE}
}

@inproceedings{gorter2024sticky,
  title={Sticky tags: Efficient and deterministic spatial memory error mitigation using persistent memory tags},
  author={Gorter, Floris and Kroes, Taddeus and Bos, Herbert and Giuffrida, Cristiano},
  booktitle={2024 IEEE Symposium on Security and Privacy (SP)},
  pages={4239--4257},
  year={2024},
  organization={IEEE}
}

@article{moghadam2024memory,
  title={Memory integrity techniques for memory-unsafe languages: A survey},
  author={Moghadam, Vahid Eftekhari and Serra, Gabriele and Aromolo, Federico and Buttazzo, Giorgio and Prinetto, Paolo},
  journal={IEEE Access},
  volume={12},
  pages={43201--43221},
  year={2024},
  publisher={IEEE}
}

@inproceedings{barberis2022branch,
  title={Branch history injection: On the effectiveness of hardware mitigations against $\{$Cross-Privilege$\}$ spectre-v2 attacks},
  author={Barberis, Enrico and Frigo, Pietro and Muench, Marius and Bos, Herbert and Giuffrida, Cristiano},
  booktitle={31st USENIX Security Symposium (USENIX Security 22)},
  pages={971--988},
  year={2022}
}

@inproceedings{koruyeh2018spectre,
  title={Spectre returns! speculation attacks using the return stack buffer},
  author={Koruyeh, Esmaeil Mohammadian and Khasawneh, Khaled N and Song, Chengyu and Abu-Ghazaleh, Nael},
  booktitle={12th USENIX Workshop on Offensive Technologies (WOOT 18)},
  year={2018}
}

@misc{golang,
  title={The Go Project},
  author={{Google}},
  url={https://go.dev/project},
  year={2025}
}

@article{brotzman2021specsafe,
  title={SpecSafe: detecting cache side channels in a speculative world},
  author={Brotzman, Robert and Zhang, Danfeng and Kandemir, Mahmut Taylan and Tan, Gang},
  journal={Proceedings of the ACM on Programming Languages},
  volume={5},
  number={OOPSLA},
  pages={1--28},
  year={2021},
  publisher={ACM New York, NY, USA}
}

@inproceedings{george2024optisan,
  title={$\{$OPTISAN$\}$: Using Multiple Spatial Error Defenses to Optimize Stack Memory Protection within a Budget},
  author={George, Rahul and Chen, Mingming and Huang, Kaiming and Qian, Zhiyun and La Porta, Thomas and Jaeger, Trent},
  booktitle={33rd USENIX Security Symposium (USENIX Security 24)},
  pages={7195--7212},
  year={2024}
}

@software{dualguard-anonymous-repo,
  author       = {Dang, Sixuan and
                  Zhou, Quan and
                  Huang, Kaiming and
                  Zhu, Jianan and
                  Gollamudi, Anitha and
                  Jaeger, Trent and
                  Zhang, Danfeng},
  title        = {DualGuard},
  month        = sep,
  year         = 2026,
  publisher    = {Zenodo},
  doi          = {10.5281/zenodo.22823571},
  url          = {https://doi.org/10.5281/zenodo.22823571},
  swhid        = {swh:1:dir:8eae3f763b8d8104947aad2d729bbf34dee8891b
                   ;origin=https://doi.org/10.5281/zenodo.22823570;vi
                   sit=swh:1:snp:e2d6eafe430e0c3287948f4a0dfe40edecb3
                   311a;anchor=swh:1:rel:848eb24eeb018baa68f4311e1889
                   50458cff5239;path=/
                  },
}

@misc{dockerhub-dualguard,
author = {Dang, Sixuan and
                  Zhou, Quan and
                  Huang, Kaiming and
                  Zhu, Jianan and
                  Gollamudi, Anitha and
                  Jaeger, Trent and
                  Zhang, Danfeng},
title = {DualGuard Docker Image},
year = {2026},
url = {https://hub.docker.com/r/dualguard/dualguard}
}

@INPROCEEDINGS {kim2025tiktag,
author = { Kim, Juhee and Park, Jinbum and Roh, Sihyeon and Chung, Jaeyoung and Lee, Youngjoo and Kim, Taesoo and Lee, Byoungyoung },
booktitle = { 2025 IEEE Symposium on Security and Privacy (SP) },
title = {{ Tiktag: Breaking ARM's Memory Tagging Extension with Speculative Execution }},
year = {2025},
volume = {},
ISSN = {},
pages = {4063-4081},
doi = {10.1109/SP61157.2025.00039},
url = {https://doi.ieeecomputersociety.org/10.1109/SP61157.2025.00039},
publisher = {IEEE Computer Society},
address = {Los Alamitos, CA, USA},
month =May}

@inproceedings{vintila2025evaluating,
  title={Evaluating the effectiveness of memory safety sanitizers},
  author={Vintila, Emanuel Q and Zieris, Philipp and Horsch, Julian},
  booktitle={2025 IEEE Symposium on Security and Privacy (SP)},
  pages={774--792},
  year={2025},
  organization={IEEE}
}

@inproceedings{kroes2018delta,
  title={Delta pointers: Buffer overflow checks without the checks},
  author={Kroes, Taddeus and Koning, Koen and Van Der Kouwe, Erik and Bos, Herbert and Giuffrida, Cristiano},
  booktitle={Proceedings of the Thirteenth EuroSys Conference},
  pages={1--14},
  year={2018}
}

@misc{noncanonical,
  title={CVE-2020-12965},
  author={{NIST}},
  url={https://nvd.nist.gov/vuln/detail/CVE-2020-12965},
  year={2022}
}

@misc{amdCVE,
  title={Transient Execution of Non-canonical Accesses},
  author={{AMD}},
  url={https://www.amd.com/en/resources/product-security/bulletin/amd-sb-1010.html},
  year={2021}
}

@inproceedings{hertogh2024leaky,
  title={Leaky address masking: Exploiting unmasked spectre gadgets with noncanonical address translation},
  author={Hertogh, Math{\'e} and Wiebing, Sander and Giuffrida, Cristiano},
  booktitle={2024 IEEE Symposium on Security and Privacy (SP)},
  pages={3773--3788},
  year={2024},
  organization={IEEE}
}

@article{musaev2021transient,
  title={Transient execution of non-canonical accesses},
  author={Musaev, Saidgani and Fetzer, Christof},
  journal={arXiv preprint},
  year={2021},
  url={https://arxiv.org/abs/2108.10771}
}

@misc{deltapointerrepo,
  url={https://github.com/vusec/deltapointers},
  year={2018}
}

@inproceedings{liljestrand2019pac,
  title={$\{$PAC$\}$ it up: Towards pointer integrity using $\{$ARM$\}$ pointer authentication},
  author={Liljestrand, Hans and Nyman, Thomas and Wang, Kui and Perez, Carlos Chinea and Ekberg, Jan-Erik and Asokan, N},
  booktitle={28th USENIX Security Symposium (USENIX Security 19)},
  pages={177--194},
  year={2019}
}

@article{partap2022memory,
  title={Memory tagging: A memory efficient design},
  author={Partap, Aditi and Boneh, Dan},
  journal={arXiv preprint arXiv:2209.00307},
  year={2022}
}

@INPROCEEDINGS {quicksafe,
author = { Blaser, Johannes and Gorter, Floris and v. Gleissenthall, Klaus and Bos, Herbert },
booktitle = { 2026 IEEE Symposium on Security and Privacy (SP) },
title = {{ QuickSafe: Targeted Hardening Against Memory Corruption }},
year = {2026},
pages = {1580-1598},
}

@inproceedings{ganjisaffar2025specasan,
  title={Specasan: Mitigating transient execution attacks using speculative address sanitization},
  author={Ganjisaffar, Saber and Koruyeh, Esmaeil Mohmmadian and Zellmer, Jason and Esfeden, Hodjat Asghari and Song, Chengyu and Abu-Ghazaleh, Nael},
  booktitle={Proceedings of the 52nd Annual International Symposium on Computer Architecture},
  pages={2032--2045},
  year={2025}
}

@inproceedings{fuchs2024safe,
  title={Safe speculation for cheri},
  author={Fuchs, Franz A and Woodruff, Jonathan and Rugg, Peter and Joannou, Alexandre and Clarke, Jessica and Baldwin, John and Davis, Brooks and Neumann, Peter G and Watson, Robert NM and Moore, Simon W},
  booktitle={2024 IEEE 42nd International Conference on Computer Design (ICCD)},
  pages={364--372},
  year={2024},
  organization={IEEE}
}

@article{woodruff2014cheri,
  title={The CHERI capability model: Revisiting RISC in an age of risk},
  author={Woodruff, Jonathan and Watson, Robert NM and Chisnall, David and Moore, Simon W and Anderson, Jonathan and Davis, Brooks and Laurie, Ben and Neumann, Peter G and Norton, Robert and Roe, Michael},
  journal={ACM SIGARCH Computer Architecture News},
  volume={42},
  number={3},
  pages={457--468},
  year={2014},
  publisher={ACM New York, NY, USA}
}

@inproceedings{reinking2021perceus,
  title={Perceus: Garbage free reference counting with reuse},
  author={Reinking, Alex and Xie, Ningning and De Moura, Leonardo and Leijen, Daan},
  booktitle={Proceedings of the 42nd ACM SIGPLAN International Conference on Programming Language Design and Implementation},
  pages={96--111},
  year={2021}
}

@inproceedings{wilson1992uniprocessor,
  title={Uniprocessor garbage collection techniques},
  author={Wilson, Paul R},
  booktitle={International Workshop on Memory Management},
  pages={1--42},
  year={1992},
  organization={Springer}
}

@inproceedings{kadir2019retpoline,
  title={Retpoline technique for mitigating spectre attack},
  author={Kadir, Mohd Fadzil Abdul and Wong, Jin Kee and Ab Wahab, Fauziah and Bharun, Ahmad Faisal Amri Abidin and Mohamed, Mohamad Afendee and Zakaria, Aznida Hayati},
  booktitle={2019 6th International Conference on Electrical and Electronics Engineering (ICEEE)},
  pages={96--101},
  year={2019},
  organization={IEEE}
}

\appendix
\subsection{Formalizing Program Semantics with Speculation}
\label{sec:langsem}

\begin{figure*}[!htb]
\begin{mathpar}
\inferrule*[right=(Fence)]{\vec{\micro}'=zero(\vec{\micro}) \and \enabled{\vec{\micro}} }
{\configthree{\mem}{\vec{\micro}}{\fencecmd} \evalsto{} \configthree{\mem}{\vec{\micro}'}{\skipcmd}}
\and
\inferrule*[right=(Asgn)]{\mathit{eval}(e,\mem)=v \and \vec{\micro}'=decr(\vec{\micro}) \and \enabled{\vec{\micro}}}
{\configthree{\mem}{\vec{\micro}}{x := e} \evalsto{\load{\location}{e}, \store{\location}{x}} \configthree{\subst{x}{v}{\mem}}{\vec{\micro}'}{\skipcmd}}
\and
\inferrule*[right=(Skip)]{\vec{\micro}\neq \emptyset \and enabled(\vec{\micro}) }
{\configthree{\mem}{\vec{\micro}}{\skipcmd} \evalsto{} \configthree{\mem}{zero(\vec{\micro})}{\skipcmd}}
\and
\inferrule*[right=(Seq1)]{\configthree{\mem}{\vec{\micro}}{\cmd_1}\evalsto{e} \configthree{\mem'}{\vec{\micro}'}{\cmd_1'}}
{\configthree{\mem}{\vec{\micro}}{\cmd_1;\cmd_2}  \evalsto{e} \configthree{\mem'}{\vec{\micro}'}{\cmd_1';\cmd_2}}
\and
\inferrule*[right=(Seq2)]{ }
{\configthree{\mem}{\vec{\micro}}{\skipcmd;\cmd_2}  \evalsto{} \configthree{\mem}{\vec{\micro}}{\cmd_2}}
\and
\inferrule*[right=(If)]{
\inferrule{}{
p=\oracle(\location).pop() \\
}
\and
c_{next}=
\begin{cases}
c_1 \text{ when } p=\true \cr
c_2 \text{ when } p=\false
\end{cases}
\\
\vec{\micro}'=decr(\vec{\micro})\cdot \config{p}{\mem}{w}{\ifcmd{b}{c_1}{c_2}}
\and \enabled{\vec{\micro}}}
{\configthree{\mem}{\vec{\micro}}{\ifcmd{b}{c_1}{c_2}}  \evalsto{\start{i}\cdot \pc{c_{next}}} \configthree{\mem}{\vec{\micro}'}{c_{next}} 
}
\and
\inferrule*[right=(While)]{ }{\configthree{\mem}{\vec{\micro}}{\whilecmd{b}{c}} \evalsto{} \configthree{\mem}{\vec{\micro}}{\ifcmd{b}{(c;\whilecmd{b}{c})}{\skipcmd}}}
\and
\inferrule*[right=(Commit)]{
\mathit{eval}(b,\mem)=p \and \enabled{\vec{\micro}'}}
{\configthree{\mem}{\config{p}{\mem'}{0}{\ifcmd{b}{c_1}{c_2}}\cdot \vec{\micro}'}{c}  \evalsto{\commit{i'}} \configthree{\mem}{\vec{\micro}'}{c} 
}
\and
\inferrule*[right=(Rollback)]{
\mathit{eval}(b,\mem) = v \neq p \and \enabled{\vec{\micro}'}
\and 
c'=
\begin{cases}
c_1 \text{ when } v=\true \cr
c_2 \text{ when } v=\false
\end{cases}
}
{\configthree{\mem}{\config{p}{\mem'}{0}{\ifcmd{b}{c_1}{c_2}}\cdot \vec{\micro}'}{c}  \evalsto{\rollback{i'}} \configthree{\mem'}{\vec{\micro}'}{c'} 
}
\end{mathpar}
\caption{Speculative-Aware semantics given an oracle $\oracle$ and a transaction length $w$.}
\label{fig:specsemantics}
\end{figure*}

\textit{Language Semantics.} Recall that to model speculative execution, we
introduce two parameters: a \emph{prediction oracle}, $\oracle$, as well as the speculative transaction's length $w$ in the language semantics. Moreover, $\oracle$ is a partial function that takes the location of a branch command and returns a queue of predicted branch outcomes (i.e., a queue of truth values). A queue is used in the oracle so that the semantics can have different outcomes for the same branch at different execution points. 

We assume that each command in the source code has a unique location $\location$. To get the next predicted outcome of the branch command with location $\location$, we use $\oracle(\location).pop()$.

\textit{Program state.} We define program state as a tuple $\configthree{\mem}{\vec{\micro}}{\cmd}$ that consists of a memory state $\mem$, a sequence of (potentially nested) speculative states $\vec{\micro}$ as well as $\cmd$, the remaining program to be executed. A speculative state $\micro$ represents a \emph{speculative transaction}~\cite{guarnieri2020spectector} that models speculative execution behavior. A single-step evaluation 
rule has the form of 
\[\configthree{\mem}{\vec{\micro}}{\mem} \evalsto{\event} \configthree{\mem'}{\vec{\micro}'}{\mem'}\] 
where event $\event$ tracks a set of security-related \emph{events}, including $\start{i}$, $\rollback{i}$ and $\commit{i}$ that denote start,
commit and rollback of the $i$th speculative transaction, as well as $\load{\location}{a[n]}$ and $\store{\location}{a[n]}$ that indicate the load and store of memory address $a[n]$ issued at program location $\location$. Since the semantics mostly follows prior work~\cite{brotzman2021specsafe}, we omit 
the evaluation rules here and formalized them in the Appendix.

We define a configuration as a tuple $\configthree{\mem}{\vec{\micro}}{\cmd}$ that consists of a memory state $\mem$, a sequence of (potentially nested) speculative states $\vec{\micro}$ as well as $\cmd$, the remaining program to be executed. As standard, memory $\mem$ is formalized as a map from variables/array elements (addresses) to their values.
A speculative state $\micro=\configfour{p}{\mem}{r}{\cmd}$ represents a \emph{speculative transaction}~\cite{guarnieri2020spectector}, which consists of the predicted boolean outcome $p$, the memory $\mem$ prior to the start of the transaction, the remaining number of instructions to be executed in the transaction $r$, and the commands $\cmd$ to be resumed in case of a rollback. Moreover, we use the following helper functions to update and check the remaining number of instructions in all transactions:
\begin{itemize}
\item decr($\vec{\micro}$): decrease the remaining instructions by one for each $\micro\in \vec{\micro}$. 

\item zero($\vec{\micro}$): set the remaining instructions to zero for each $\micro\in \vec{\micro}$. 

\item enabled($\vec{\micro}$): returns true if $r\neq 0$ for any $\micro\in \vec{\micro}$ or when $\vec{\micro}=\emptyset$. 
\end{itemize}

Given an oracle $\oracle$ and a transaction length $w$, each evaluation rule (summarized in Figure~\ref{fig:specsemantics}) has the form of $\configthree{\mem}{\vec{\micro}}{\cmd} \evalsto{e} \configthree{\mem'}{\vec{\micro}'}{\cmd'}$ where event $e$ tracks the memory locations being accessed and program instructions being executed; we use $\pc{\location}$ to denote the latter. Note that the given $\oracle$ and $w$, the semantics is deterministic: rules COMMIT and ROLLBACK apply when one transaction ended, while the other rules require $\enabled{\vec{\micro}}$; rule COMMIT requires $\mathit{eval}(b,\mem) = p$ and ROLLBACK requires $\mathit{eval}(b,\mem) \neq p$.


Given a program $\cmd$, an oracle $\oracle$ and transaction length $w$, we use $\trace^\oracle_{\langle \mem, w \rangle}(\cmd)$  to denote the sequence of emitted events starting from the initial configuration $\configthree{\mem}{\emptyset}{\cmd}$ according to the semantics given $\oracle$ and $w$. For example, let $\cmd$ be the following code, where $\oracle(1)=\true$, $w$=1, $\mem(\codevar{i}) \ge \mem(\codevar{A\_size})$.

\begin{lstlisting}
if (i < A_size)
    y = A[i];
\end{lstlisting}

The evaluation steps is as follows according to the semantics:
\begin{align*}
  &\configthree{\mem}{\emptyset}{\cmd} \\
\evalsto{\start{1}\cdot \pc{2}} &\configthree{\mem}{\config{\true}{\mem}{1}{\cmd}}{\texttt{y=A[i]}} \\
\evalsto{\load{2}{A[i]}, \store{2}{y}} &\configthree{\mem'}{\config{\true}{\mem}{0}{\cmd}}{\skipcmd} \\
\evalsto{\rollback{1}} &\configthree{\mem}{\emptyset}{\skipcmd} \\
\end{align*}

Hence, $\trace^\oracle_{\langle \mem, w \rangle}(\cmd)$ generates an event trace $\start{1}$, $\pc{2}$, $\load{2}{A[i]}$, $\store{2}{y}$, $\rollback{1}$ for the concrete execution. 

We note that our interpretation of $\oracle$ and $w$ over-approximates the ones in~\cite{guarnieri2020spectector}, which assumes a partial function that takes as input a program, a branching history and a branch instruction, and that returns as output  the predicted branch outcome and transaction length. Our formulation  considers all possible outcomes of the semantics in~\cite{guarnieri2020spectector}.

\section{Value Range Analysis}

The following are the data flow equations we used in our value range analysis.

\begin{equation}
\label{equ:vr}
    \begin{aligned}
        In(B) &= \bigcup_{B' \in pred(B)} Out(B') \\
        Out(B) &= (In(B) - kill(B)) \cup gen(B)
    \end{aligned}
\end{equation}

We use a forward-may data flow analysis to compute the range of each variable in the program. Equation \Cref{equ:vr} shows the data flow equations used by the value range analysis. $In(B)$ is the value range information at the entry point of basic block B, which is computed by taking the union of range at the exit point of all its predecessors, i.e., $Out(B')$. Specifically, if a variable's range information can be found in every predecessor, the analysis takes the union of all the range. If a variable's range does not exist in at least one of the predecessors, meaning the range cannot be decided at this point, so the resulting range is infinite.

$Out(B)$ is the range at the exit point of basic block $B$. It is computed by deducting the kill set $kill(B)$ and then union the generation set $gen(B)$. The kill set $kill(B)$ and the generation set $gen(B)$ contains variables that are defined in basic block B. A variable in $kill(B)$ and $gen(B)$ might be in $In(B)$ already, which indicates this variable is updated in basic block B. By deducting $kill(B)$ and union $gen(B)$, the analysis is able to first remove the old range in $In(B)$ and then add the newly updated range.

The following is an example to illustrate the impact of path-sensitivity on the two phases.

\begin{lstlisting}
int func(int i) {
    int A[8];
    if (i < 8) {
        return A[i];
    }
    return -1;
}
\end{lstlisting}

With path-sensitivity, the path condition that has to be satisfied to reach line 4 is $i < 8$. Under sequential execution, the memory access at line 4 is memory safe, and it will be marked as safe by the path-sensitive analysis. Under speculative execution, however, an attacker can train the branch predictor to take true branch when $i \ge 8$, which makes the path condition no longer holds. We took the heuristic that the attacker ignoring the branch condition is equivalent to the analysis dropping path condition when constraining the value range. The path-insensitive analysis would reject the access at line 4.

\begin{figure*}[t]
    \centering
    \includegraphics[width=0.95\linewidth]{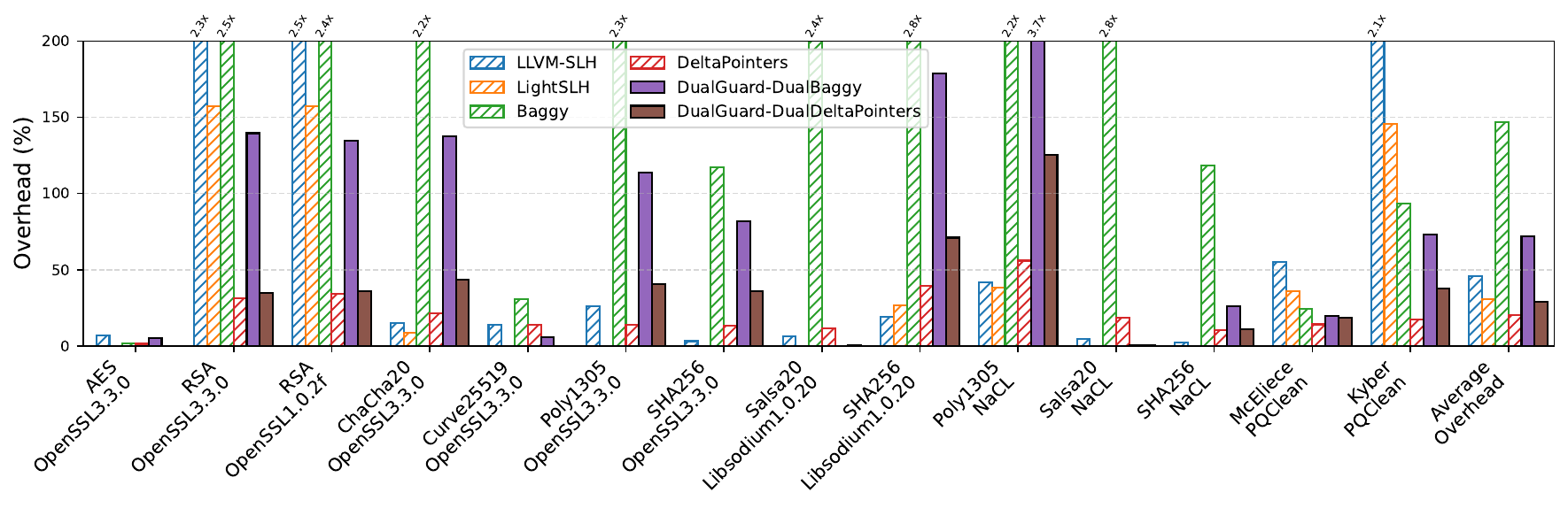}
    \caption{Comparison between partial defenses and \framework for cryptographic libraries}
    \label{fig:partial-overhead}
\end{figure*}

\begin{figure*}[t]
    \centering
    \includegraphics[width=0.95\linewidth]{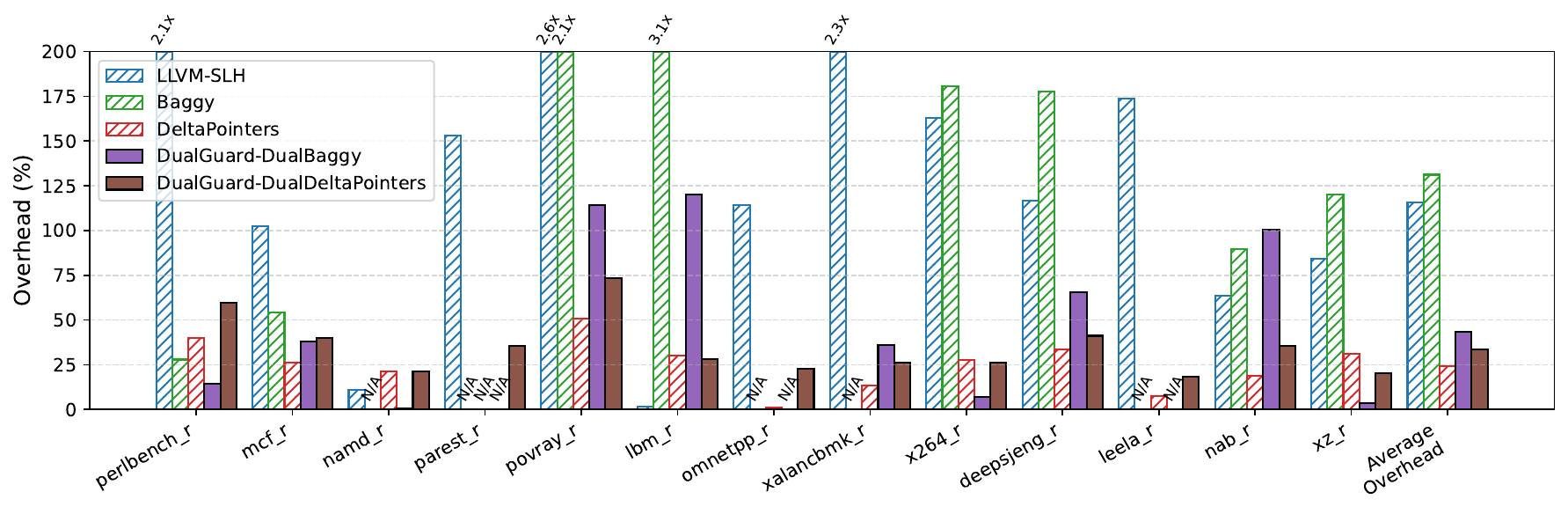}
    \caption{Comparison between partial defenses and \framework for SPEC benchmarks}
    \label{fig:partial-overhead-spec}
\end{figure*}

\subsection{Comparison with Partial Protections}

The detailed comparison with partial protections for each benchmark is shown in~\Cref{fig:partial-overhead} and~\Cref{fig:partial-overhead-spec} for cryptographic libraries and SPEC benchmark suite, respectively.

\textbf{Sequential defense.} \changes{Across the SPEC benchmarks, Baggy incurs overhead ranging from 27.9\% on perlbench\_r to 410\% on lbm\_r, with an average of 131\%. In comparison, \frameworkbaggy incurs an average overhead of only 44\%, corresponding to a 66\% reduction relative to Baggy. Delta Pointers achieves lower overhead than Baggy, ranging from 1.3\% on omnetpp\_r to 51\% on povray\_r. \frameworkdp demonstrates comparable or slightly higher overhead due to the extra instructions for dual-safety. Overall, \frameworkdp achieves overhead that is lower than or comparable to that of Delta Pointers, while providing a stronger security guarantee.}

Across the cryptographic library benchmarks, Baggy exhibits substantial runtime overhead, ranging from 1.8\% for AES in OpenSSL 3.3.0 to as high as 280\% for SHA-256 in Libsodium 1.0.20. Despite providing stronger security guarantees, \frameworkbaggy demonstrates consistently lower overhead across most benchmarks. The primary exception is Poly1305 in NaCl, where Baggy incurs lower overhead. Overall, Baggy introduces an average overhead of 146.6\%, whereas \frameworkbaggy incurs an average overhead of 71.9\%, representing a relative reduction of approximately 51\%. Delta Pointers observes a lower overhead than Baggy, ranging from 1.3\% to 56\%. \frameworkdp demonstrates comparable or slightly higher overhead due to the extra instructions for dual-safety. On average, Delta Pointer and \frameworkdp incur overhead of 20.4\% and 28.9\%, respectively, where \framework shows a slight overhead increase. 

\textbf{Speculative defense.} \changes{We were unable to evaluate LightSLH on the SPEC benchmarks due to its limited instruction support and scalability limitations. LLVM's SLH incurs overhead ranging from 1\% to 264\%, with an average of 116\%. Compared with LLVM's SLH, \frameworkbaggy achieves lower overhead on all benchmarks except lbm\_r and nab\_r, reducing overhead by 33\% on average. Similarly, \frameworkdp incurs slightly higher overhead only on namd\_r and lbm\_r, while achieving an average overhead reduction of 37.9\%.
}

For cryptographic libraries, \frameworkbaggy incur average overhead of 71.9\%, whereas LLVM’s SLH and LightSLH incur overheads of 45.6\% and 30.6\%, respectively. This higher overhead primarily stems from the fact that \frameworkbaggy instrumentations consist of memory loads for the bounds table. Notably, despite its stronger safety, \frameworkbaggy achieves lower overhead in 7 evaluated benchmarks compared to LLVM’s SLH, and in 4 benchmark compared to LightSLH. \frameworkdp, on the other hand, observes an overhead of only 28.9\%, which is lower than both SLH variants while providing dual-safety thanks to the in-pointer metadata storage. 

\begin{figure*}[t]
    \centering
    \includegraphics[width=0.9\linewidth]{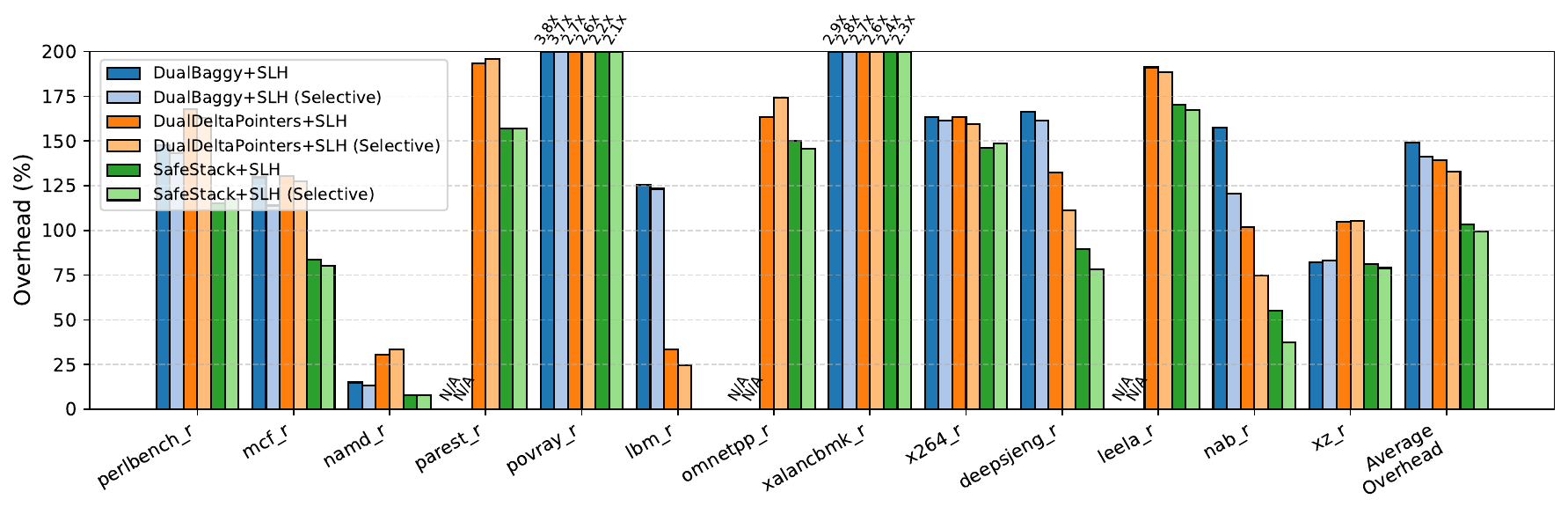}
    \caption{The effect of selective protection on SPEC benchmarks}
    \label{fig:selective-spec}
\end{figure*}

\begin{figure*}[t]
    \centering
    \includegraphics[width=0.95\linewidth]{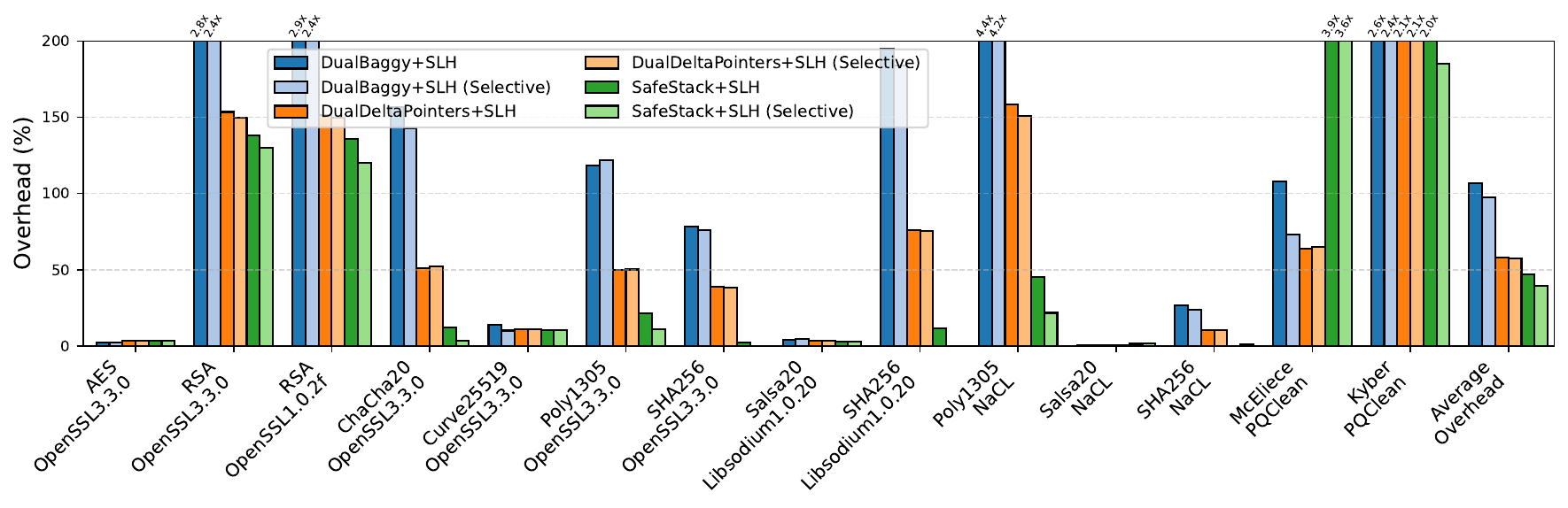}
    \caption{The effect of selective protection on cryptographic libraries}
    \label{fig:selective-crypto}
\end{figure*}

\subsection{Case Study}

To better understand how \framework's configurations achieve lower overheads than naive compositions, we provide a real-world example in the big number library (bn) of OpenSSL 3.3.0. The big numbers in OpenSSL are managed in an internal stack. \framework analysis reports a potential out-of-bounds store in the stack push function below.

\begin{lstlisting}[
    linebackgroundcolor={%
    \ifnum\value{lstnumber}=17
            \color{red!30}
    \fi
    }
]
static int BN_STACK_push(BN_STACK *st, unsigned int idx)
{
    if (st->depth == st->size) {
        /* Need to expand */
        unsigned int newsize =
            st->size ? (st->size * 3 / 2) : BN_CTX_START_FRAMES;
        unsigned int *newitems;

        if ((newitems = OPENSSL_malloc(sizeof(*newitems) * newsize)) == NULL)
            return 0;
        if (st->depth)
            memcpy(newitems, st->indexes, sizeof(*newitems) * st->depth);
        OPENSSL_free(st->indexes);
        st->indexes = newitems;
        st->size = newsize;
    }
    st->indexes[(st->depth)++] = idx;
    return 1;
}
\end{lstlisting}

In this function, the stack is resized and reallocated if the current stack cannot hold more entries. The new size computation at Line 6 can potentially overflow the integer size and become smaller (\codevar{newsize} $<$ \codevar{st->depth}), making the allocation size at Line 9 smaller than the actual required size. Finally, at Line 17, the access to the index \codevar{st->depth} could write into memory regions outside of the allocation bounds. This pattern of integer overflow followed by out-of-bounds write can be found in many high-risk vulnerabilities (e.g., CVE-2025-39967, CVE-2026-5121, and CVE-2026-5444).

In a naive composition of \baselinebaggy, the store operation at Line 17 will be transformed to the following code, with the Baggy instrumentation marked as blue and SLH instrumentation marked as yellow.

\begin{lstlisting}[
    linebackgroundcolor={%
    \ifnum\value{lstnumber}=4
            \color{blue!30}
    \fi
    \ifnum\value{lstnumber}=5
            \color{blue!30}
    \fi
    \ifnum\value{lstnumber}=6
            \color{blue!30}
    \fi
    \ifnum\value{lstnumber}=8
            \color{blue!30}
    \fi
    \ifnum\value{lstnumber}=7
            \color{yellow!30}
    \fi
    \ifnum\value{lstnumber}=9
            \color{yellow!30}
    \fi
    \ifnum\value{lstnumber}=10
            \color{yellow!30}
    \fi
    \ifnum\value{lstnumber}=11
            \color{yellow!30}
    \fi
    }
]
depth = st->depth;
base = st->indexes;
addr = &base[depth];
log_size = size_table[base >> log_slot_size];
is_oob = (addr XOR base) >> log_size != 0
if(is_oob) {
    mis = is_oob ? 0 : -1
    addr = slow_path(base, addr);
} else 
    mis = !is_oob ? 0 : -1
hardened_addr = addr | mis;
store(idx, hardened_addr);
st->depth++;
\end{lstlisting}

Baggy introduces a branch to check if the index is out of the allocation bounds (i.e., \codevar{newsize}), which ensures the store operation stays inside bounds in sequential execution. Then, to preserve the semantics of the program before hardening, SLH inserts an else branch to keep track of the mis-speculation flag \codevar{mis}. Upon storing, the address is hardened by \codevar{mis}. With the layered instrumentation, when reaching Line 12, the transformed code will correctly set the MSB if \codevar{depth} is greater than \codevar{newsize}. That said, the store operation still relies on the hardware check to reject an address where MSB is set, leaving the complex composition superfluous.

In contrast, \frameworkbaggy instrument Line 17 as follows. By checking the bounds in a branchless manner and removing the reliance on hardware check, \dualbaggy instrumentation is more efficient and robust against speculative attacks at the same time.

\begin{lstlisting}[
    linebackgroundcolor={%
    \ifnum\value{lstnumber}=4
            \color{blue!30}
    \fi
    \ifnum\value{lstnumber}=5
            \color{blue!30}
    \fi
    \ifnum\value{lstnumber}=6
            \color{blue!30}
    \fi
    }
]
depth = st->depth;
base = st->indexes;
addr = &base[depth];
log_size = size_table[base >> log_slot_size];
is_oob = (addr XOR base) >> log_size != 0
final_addr = select(is_oob, -1, addr);
store(idx, hardened_addr);
st->depth++;
\end{lstlisting}

\subsection{Selective Defense Overhead}

The detailed comparison of selective defenses for each benchmark is shown in~\Cref{fig:selective-spec} and~\Cref{fig:selective-crypto}.

\subsection{Analysis Running Time}

The running time for our analysis is shown in \Cref{tab:runnningtime}.

\begin{table}
\centering
\caption{Analysis running time}
\label{tab:runnningtime}
\begin{tabular}{ll}
\toprule
Benchmark               & Running Time (s)  \\
\midrule
perlbench\_r & 8254.01  \\
mcf\_r       & 6.77     \\
namd\_r      & 24.55    \\
parest\_r    & 604.2    \\
povray\_r    & 1309.28  \\
lbm\_r       & 1.42     \\
omnetpp\_r   & 180.49   \\
xalancbmk\_r & 51.23    \\
x264\_r      & 937.43   \\
deepsjeng\_r & 156.22   \\
leela\_r     & 7.9      \\
nab\_r       & 149.82   \\
xz\_r        & 2.98     \\
\midrule
\midrule
OpenSSL3.3.0/AES        & 0.04              \\
OpenSSL3.3.0/RSA        & 6.69              \\
OpenSSL1.0.2f/RSA       & 7.54              \\
OpenSSL3.3.0/ChaCha20   & 0.07              \\
OpenSSL3.3.0/Curve25519 & 0.95              \\
OpenSSL3.3.0/Poly1305   & 0.03              \\
OpenSSL3.3.0/SHA256     & 0.28              \\
Libsodium1.0.20/Salsa20 & 0.05              \\
Libsodium1.0.20/SHA256  & 0.17              \\
NaCL/Poly1305           & 0.2               \\
NaCL/Salsa20            & 0.03              \\
NaCL/SHA256             & 0.18              \\
PQClean/McEliece        & 0.16               \\
PQClean/Kyber           & 0.99              \\
\bottomrule
\end{tabular}
\end{table}



\subsection{Artifact}

This section details how to build and run \framework. The tool chain can be found at https://doi.org/10.5281/zenodo.22823571. The repositories include all the source code of \framework, and benchmarks except SPEC CPU 2017 benchmark suite that cannot be shared due to copyright issue. \framework can be either run by using the pre-built Docker image~\cite{dockerhub-dualguard} or by building from source~\cite{dualguard-anonymous-repo}. We recommend using the pre-built Docker image to ease the burden on setting up the tool chain.

\subsubsection{Using Pre-built \framework with Docker}
Make sure the test machine has Docker installed. We provide the Docker image on Dockerhub~\cite{dockerhub-dualguard}. Get the container running with the following commands.

\begin{lstlisting}[language=bash]
# Pull Docker image from DockerHub
docker pull dualguard/dualguard:latest

# Run a container with the image
docker run --name dualguard -dit dualguard/dualguard

# Get into the container's bash
docker exec -it dualguard bash
\end{lstlisting}

Inside the container, the artifact is located at \texttt{/root/dualguard}. \framework has been built in the image, but SPEC CPU 2017 benchmark suite has to be installed manually due to copyright reasons. Please install the benchmark suite according to its instructions.

After the benchmark suite is successfully installed, edit \texttt{SPEC\_PATH} to point to the installation directory, and run the setup script.

\begin{lstlisting}[language=bash]
# Set up benchmark harness for SPEC
bash ./setup.sh
\end{lstlisting}

\subsubsection{Building \framework from Source} We have tested building the system under Ubuntu 24.04, which provides the best compatibility of software used. Please have 30 GB of disk space available on your system. The instructions in this documentation assume the user is running this version of Ubuntu Linux.

First, install the dependencies.

\begin{lstlisting}[language=bash]
sudo apt update
sudo apt install git build-essential python3 ninja-build wget cmake libtinfo-dev libtinfo6
\end{lstlisting}

Then, download the source code from the repo~\cite{dualguard-anonymous-repo}, and decompress the source code. 

Before running the building script, please install SPEC CPU 2017 benchmark suite and edit \texttt{SPEC\_PATH} to point to the installation directory. Finally, navigate into the root directory of the source code and run the building script. 

\begin{lstlisting}[language=bash]
bash ./build.sh
\end{lstlisting}







Run the benchmarks with the following commands accordingly.

\begin{lstlisting}[language=bash]
# run SPEC benchmarks
./run.sh spec

# run LightSLH benchmarks
./run.sh lightslh

# run both
./run.sh
\end{lstlisting}

If SPEC CPU 2017 benchmark suite is not available, run \codevar{run.sh} lightslh to run the LightSLH benchmarks only.

Using the following commands to visualizing the results.

\begin{lstlisting}[language=bash]
# This will create Figure 5 and 10 in the paper
cd $SPEC_ROOT/scripts
python3 plots.py

# Produce Figure 6, 9 and 11 in the paper
cd LightSLH-benchmarks/bench/wrapper
python3 plots.py
\end{lstlisting}

\end{document}